\documentclass[reqno,12pt,a4paper]{amsart}

\usepackage{amsaddr}
\usepackage[draft]{fixme}
\usepackage{mathrsfs}
\usepackage{amssymb}
\usepackage{amsfonts}
\usepackage{latexsym}
\usepackage{amsthm}
\usepackage{enumerate}
\usepackage{enumitem}
\usepackage{times}

\usepackage{graphicx}
\def\lb{\label}

\newcommand{\Lan}{\mathrm{Lan}}
\newcommand{\er}[1]{\textrm{(\ref{#1})}}

\begin{document}

%%%%%%%%%% Some definitions %%%%%%%%%%

%%%%%%%% Equations, theorems %%%%%%%%%
\renewcommand{\theequation}{\arabic{section}.\arabic{equation}}
\theoremstyle{plain}
\newtheorem{theorem}{\bf Theorem}[section]
\newtheorem{lemma}[theorem]{\bf Lemma}
\newtheorem{corollary}[theorem]{\bf Corollary}
\newtheorem{proposition}[theorem]{\bf Proposition}
\newtheorem{definition}[theorem]{\bf Definition}
\theoremstyle{remark}
\newtheorem{remark}[theorem]{\bf Remark}
\newtheorem{remarks}[theorem]{\bf Remarks}
\newtheorem{example}[theorem]{\bf Example}

 \numberwithin{equation}{section}

%%%%% Alphabet %%%%%
\def\a{\alpha}  \def\cA{{\mathcal A}}     \def\bA{{\bf A}}  \def\mA{{\mathscr A}}
\def\b{\beta}   \def\cB{{\mathcal B}}     \def\bB{{\bf B}}  \def\mB{{\mathscr B}}
\def\g{\gamma}  \def\cC{{\mathcal C}}     \def\bC{{\bf C}}  \def\mC{{\mathscr C}}
\def\G{\Gamma}  \def\cD{{\mathcal D}}     \def\bD{{\bf D}}  \def\mD{{\mathscr D}}
\def\d{\delta}  \def\cE{{\mathcal E}}     \def\bE{{\bf E}}  \def\mE{{\mathscr E}}
\def\D{\Delta}  \def\cF{{\mathcal F}}     \def\bF{{\bf F}}  \def\mF{{\mathscr F}}
\def\c{\chi}    \def\cG{{\mathcal G}}     \def\bG{{\bf G}}  \def\mG{{\mathscr G}}
\def\z{\zeta}   \def\cH{{\mathcal H}}     \def\bH{{\bf H}}  \def\mH{{\mathscr H}}
\def\e{\eta}    \def\cI{{\mathcal I}}     \def\bI{{\bf I}}  \def\mI{{\mathscr I}}
\def\p{\psi}    \def\cJ{{\mathcal J}}     \def\bJ{{\bf J}}  \def\mJ{{\mathscr J}}
\def\vT{\Theta} \def\cK{{\mathcal K}}     \def\bK{{\bf K}}  \def\mK{{\mathscr K}}
\def\k{\kappa}  \def\cL{{\mathcal L}}     \def\bL{{\bf L}}  \def\mL{{\mathscr L}}
\def\l{\lambda} \def\cM{{\mathcal M}}     \def\bM{{\bf M}}  \def\mM{{\mathscr M}}
\def\L{\Lambda} \def\cN{{\mathcal N}}     \def\bN{{\bf N}}  \def\mN{{\mathscr N}}
\def\m{\mu}     \def\cO{{\mathcal O}}     \def\bO{{\bf O}}  \def\mO{{\mathscr O}}
\def\n{\nu}     \def\cP{{\mathcal P}}     \def\bP{{\bf P}}  \def\mP{{\mathscr P}}
\def\r{\rho}    \def\cQ{{\mathcal Q}}     \def\bQ{{\bf Q}}  \def\mQ{{\mathscr Q}}
\def\s{\sigma}  \def\cR{{\mathcal R}}     \def\bR{{\bf R}}  \def\mR{{\mathscr R}}
\def\S{\Sigma}  \def\cS{{\mathcal S}}     \def\bS{{\bf S}}  \def\mS{{\mathscr S}}
\def\t{\tau}    \def\cT{{\mathcal T}}     \def\bT{{\bf T}}  \def\mT{{\mathscr T}}
\def\f{\phi}    \def\cU{{\mathcal U}}     \def\bU{{\bf U}}  \def\mU{{\mathscr U}}
\def\F{\Phi}    \def\cV{{\mathcal V}}     \def\bV{{\bf V}}  \def\mV{{\mathscr V}}
\def\P{\Psi}    \def\cW{{\mathcal W}}     \def\bW{{\bf W}}  \def\mW{{\mathscr W}}
\def\om{\omega} \def\cX{{\mathcal X}}     \def\bX{{\bf X}}  \def\mX{{\mathscr X}}
\def\x{\xi}     \def\cY{{\mathcal Y}}     \def\bY{{\bf Y}}  \def\mY{{\mathscr Y}}
\def\X{\Xi}     \def\cZ{{\mathcal Z}}     \def\bZ{{\bf Z}}  \def\mZ{{\mathscr Z}}
\def\Om{\Omega}

\newcommand{\gA}{\mathfrak{A}}
\newcommand{\gB}{\mathfrak{B}}
\newcommand{\gC}{\mathfrak{C}}
\newcommand{\gD}{\mathfrak{D}}
\newcommand{\gE}{\mathfrak{E}}
\newcommand{\gF}{\mathfrak{F}}
\newcommand{\gG}{\mathfrak{G}}
\newcommand{\gH}{\mathfrak{H}}
\newcommand{\gI}{\mathfrak{I}}
\newcommand{\gJ}{\mathfrak{J}}
\newcommand{\gK}{\mathfrak{K}}
\newcommand{\gL}{\mathfrak{L}}
\newcommand{\gM}{\mathfrak{M}}
\newcommand{\gN}{\mathfrak{N}}
\newcommand{\gO}{\mathfrak{O}}
\newcommand{\gP}{\mathfrak{P}}
\newcommand{\gR}{\mathfrak{R}}
\newcommand{\gS}{\mathfrak{S}}
\newcommand{\gT}{\mathfrak{T}}
\newcommand{\gU}{\mathfrak{U}}
\newcommand{\gV}{\mathfrak{V}}
\newcommand{\gW}{\mathfrak{W}}
\newcommand{\gX}{\mathfrak{X}}
\newcommand{\gY}{\mathfrak{Y}}
\newcommand{\gZ}{\mathfrak{Z}}

\def\ve{\varepsilon}   \def\vt{\vartheta}    \def\vp{\varphi}    \def\vk{\varkappa}

\def\Z{{\mathbb Z}}    \def\R{{\mathbb R}}   \def\C{{\mathbb C}}
\def\T{{\mathbb T}}    \def\N{{\mathbb N}}   \def\dD{{\mathbb D}}
\def\B{{\mathbb B}}

%%%%% Arrows %%%%%

\def\la{\leftarrow}              \def\ra{\rightarrow}      \def\Ra{\Rightarrow}
\def\ua{\uparrow}                \def\da{\downarrow}
\def\lra{\leftrightarrow}        \def\Lra{\Leftrightarrow}

%%%%% Typography %%%%%

\def\lt{\biggl}                  \def\rt{\biggr}
\def\ol{\overline}               \def\wt{\widetilde}
\def\no{\noindent}

%%%%% Math signs %%%%%

\let\ge\geqslant                 \let\le\leqslant
\def\lan{\langle}                \def\ran{\rangle}
\def\/{\over}                    \def\iy{\infty}
\def\sm{\setminus}               \def\es{\emptyset}
\def\ss{\subset}                 \def\ts{\times}
\def\pa{\partial}                \def\os{\oplus}
\def\ev{\equiv}
\def\iint{\int\!\!\!\int}        \def\iintt{\mathop{\int\!\!\int\!\!\dots\!\!\int}\limits}
\def\el2{\ell^{\,2}}             \def\1{1\!\!1}
\def\sh{\sharp}
\def\wh{\widehat}
%%%%% Math operations %%%%%

\def\where{\mathop{\mathrm{where}}\nolimits}
\def\as{\mathop{\mathrm{as}}\nolimits}
\def\Area{\mathop{\mathrm{Area}}\nolimits}
\def\arg{\mathop{\mathrm{arg}}\nolimits}
\def\const{\mathop{\mathrm{const}}\nolimits}
\def\det{\mathop{\mathrm{det}}\nolimits}
\def\diag{\mathop{\mathrm{diag}}\nolimits}
\def\diam{\mathop{\mathrm{diam}}\nolimits}
\def\dim{\mathop{\mathrm{dim}}\nolimits}
\def\dist{\mathop{\mathrm{dist}}\nolimits}
\def\Im{\mathop{\mathrm{Im}}\nolimits}
\def\Iso{\mathop{\mathrm{Iso}}\nolimits}
\def\Ker{\mathop{\mathrm{Ker}}\nolimits}
\def\Lip{\mathop{\mathrm{Lip}}\nolimits}
\def\rank{\mathop{\mathrm{rank}}\limits}
\def\Ran{\mathop{\mathrm{Ran}}\nolimits}
\def\Re{\mathop{\mathrm{Re}}\nolimits}
\def\Res{\mathop{\mathrm{Res}}\nolimits}
\def\res{\mathop{\mathrm{res}}\limits}
\def\sign{\mathop{\mathrm{sign}}\nolimits}
\def\span{\mathop{\mathrm{span}}\nolimits}
\def\supp{\mathop{\mathrm{supp}}\nolimits}
\def\Tr{\mathop{\mathrm{Tr}}\nolimits}
\def\BBox{\hspace{1mm}\vrule height6pt width5.5pt depth0pt \hspace{6pt}}

%%%%%%%%%%%%% specialities %%%%%%%%%%%%%%

\newcommand{\BS}{\mathrm{BS}}
\newcommand{\pp}{\mathrm{pp}}
\newcommand{\Lips}{\mathrm{Lip}}
\newcommand{\tJ}{\widetilde{J}}
\newcommand{\one}{{\bf 1}}

\newcommand\nh[2]{\widehat{#1}\vphantom{#1}^{(#2)}}
%{{\mathop{#1}\limits^\wedge}\vphantom{#1}^{(#2)}}
\def\dia{\diamond}

\def\Oplus{\bigoplus\nolimits}

\newcommand{\slim}{\mathrm{s}-\lim}

%%%%%%%%%%% End of definitions %%%%%%%%%%

%%%%% OLD OLD OLD

\def\qqq{\qquad}
\def\qq{\quad}
\let\ge\geqslant
\let\le\leqslant
\let\geq\geqslant
\let\leq\leqslant
\newcommand{\ca}{\begin{cases}}
\newcommand{\ac}{\end{cases}}
\newcommand{\ma}{\begin{pmatrix}}
\newcommand{\am}{\end{pmatrix}}
\renewcommand{\[}{\begin{equation}}
\renewcommand{\]}{\end{equation}}
\def\bu{\bullet}

\title[{Weighted estimates for the Laplacian on the cubic lattice}]
{Weighted estimates for the Laplacian on the cubic lattice}

\date{\today}

\author[Evgeny Korotyaev]{Evgeny L. Korotyaev}
\address{Saint Petersburg State University, St. Petersburg, Russia}
\email{korotyaev@gmail.com,  e.korotyaev@spbu.ru}

\author[Jacob  Schach  M{\o}ller]{Jacob  Schach  M{\o}ller}
\address{Department of Mathematics, Aarhus University, Denmark.}
\email{jacob@math.au.dk}


\begin{abstract}
We consider the discrete Laplacian $\D$ on the cubic lattice $\Z^d$, and
deduce estimates on the group $e^{it\D}$ and the resolvent
$(\D-z)^{-1}$, weighted by $\ell^q(\Z^d)$-weights  for suitable
$q\geq 2$. We apply the obtained results to discrete Schr\"odinger
operators in dimension $d\geq 3$ with potentials from $\ell^p(\Z^d)$ with suitable $p\geq 1$.
\end{abstract}

\maketitle

\section{Introduction and main results}

\subsection{Introduction}
We consider the Schr\"odinger operator $H$ acting in
$\ell^2(\Z^{d}),d\ge 3$ and given by
\[\label{Hamil}
{H}=\D+V,
\]
where ${\D}$ is the discrete Laplacian on $\Z^d$ defined  by
\[
\lb{dL} \big(\D f \big)(n)=\frac{1}{2}\sum_{j=1}^{d}\bigl( f_{n+
e_{j}} + f_{n-e_{j}}\bigr),
\]
for $f =(f_n)_{n\in\Z^d} \in \ell^{2}(\Z^d)$. Here $ e_{1} =
(1,0,\cdots,0), \cdots, e_{d} = (0,\cdots,0,1) $
 is the standard basis of $\Z^d$. The spectrum of $\D$ is absolutely continuous and
 $\s(\D)=\s_{\textup{ac}}(\D)=[-d,d]$, and the threshold set -- critical energies of $\Delta$ in its momentum representation -- is $\tau(H) = \tau(\Delta) = (2\Z+d)\cap [-d,d]$.

The operator  $V$ is the operator of multiplication by a real-valued potential function, $V = (V_n)_{n\in\Z^d}$.
%potential acting by
%%
%\begin{equation}
%(Vf)(n)=V_nf_n, \qqq n\in \Z^d.
%\end{equation}
%%
The potentials we work with will always satisfy that $\lim_{n\to\infty} V_n=0$,
%assume that the potential $V$ satisfies  the following condition:
%%
%\begin{equation}
%\lb{Vcx} V\in \ell^{p}(\Z^d),\qqq 1\le p<\iy.
%\end{equation}
and consequently the
essential spectrum of the Schr\"odinger operator $H$ on
$\ell^2(\Z^d)$ is $\s_{\mathrm{ess}}(H)=[-d,d]$.
However, this does not exclude appearance of eigenvalues and singular
continuous spectrum in the interval $[-d,d]$. 

Before describing our results, we recall some well known classic results by Kato for the continuous case, established in the famous paper \cite{Ka66}. Kato
considered the Laplacian $\D$ acting in the space $L^2(\R^d)$ for $d\ge
2$. He proved the following estimates:

\begin{enumerate}[label = \textup{(\roman*)}]
\item  Let $d\geq 2$ and $q>2$. Then for all $t\in \R\sm \{0\}$ and
$u\in L^{q}(\R^d)$, we have
\begin{equation}
\label{eLR} \|u e^{it\D}u\|\le C_{d,q}|t|^{-{d\/q}}\|u\|_q^2,
\end{equation}
for some constant $C_{d,q}$ depending on $d$ and $q$ only.
\item Let $\ve>0$ and $d>2+\ve$. For any $u\in L^{d -\ve}(\R^d)\cap L^{d + \ve}(\R^d)$, we have
\begin{equation}
\label{eRR}
\forall \ \l\in \C\sm [0,\iy):\qquad
 \|u (-\D-\l)^{-1}u\|\le
C_{d,\ve}\bigl(\|u\|_{d-\ve}^2+\|u\|_{d+\ve}^2\bigr),
\end{equation}
for some constant $C_{d,\ve}$ depending on $d$ and $\ve$ only.
\end{enumerate}

\medskip

These estimates have a lot of applications in spectral theory, see
\cite{RS78}. Note that \er{eLR} is a simple example of a dispersive
(or Strichartz) estimate. Dispersive estimates are
very useful in the theory of linear and non-linear partial
differential equations, see \cite{Sc07} and references therein. Note
that the estimate \er{eLR} implies \er{eRR} and  that the
operator-valued function $u (-\D-\l)^{-1}u$ is analytic in $\C\sm
[0,\iy)$ and uniformly H\"older up to the boundary.

In the present paper we prove a dispersive estimate of the type \er{eLR}
for the discrete Laplacian \eqref{dL}, and we show that by replacing the $q$-norm with a 
weighted $q$-norm on the right-hand side one may improve the time-decay. This is the content of Theorem~\ref{T1} below. 
Secondly, we establish resolvent estimates of the type \eqref{eLR} that are better than what is implied
by our dispersive estimates. See Theorem~\ref{T2}. That is, unlike the continuous case studied by Kato, one does not get good resolvent bounds merely by integrating a dispersive estimate of the type \eqref{eLR}. Instead we analyze and exploit the pointwise decay of the summation kernel for the free resolvent. The starting point for our analysis is a representation of the summation kernel of the propagator in terms of a product of Bessel functions. Our estimates then follow from a careful use of recent optimal estimates on Bessel functions by Krasikov \cite{Kr14} and Landau \cite{La00}. Finally, in Theorem~\ref{T3}, we deduce consequences for the spectral and scattering theory for the Hamiltonian \eqref{Hamil} with potentials from a suitable $\ell^p(\Z^d)$ space. The proof of this last theorem revolves around Birman-Schwinger type arguments.

\medskip

For the discrete Schr\"odinger operators on the cubic lattice, most results were
obtained for uniformly decaying potentials for the $\Z^1$ case, see
for example \cite{DHKS,IJ15,KS03,T89}. 
For real-valued finitely supported potentials in $\Z^d$, Kopylova \cite{Ko10} has established dispersive estimates, 
Shaban and Vainberg \cite{SV01} as well as Ando, Isozaki and Morioka \cite[Thm.~7.7]{AIM16}, studied Limiting Absorption away 
from thresholds, and Isozaki and Morioka \cite[Thm.~2.1]{IM14} proved that the
point-spectrum of $H$ on the interval $(-d,d)$ is absent. Note that
in \cite{HSSS12} the authors gave an example in dimension $d\geq 5$ of an embedded eigenvalue at
the endpoint $\{\pm d\}$. Recently, Hayashi, Higuchi, Nomura, and Ogurisu computed the number of discrete eigenvalues for finitely supported potential \cite{HHNO16}.

 For Schr\"odinger operators with decreasing
potentials on the lattice $\Z^d$, Boutet de
Monvel and Sahbani \cite{BS99} used Mourre's method to prove absence of singular continuous spectrum and local finiteness of eigenvalues away from threshold energies $\tau(H)$, a technique revisited by Isozaki and Korotyaev \cite{IK12}, who also studied the direct and the inverse scattering problem as well as trace formulae.

For $\ell^{d/2}(\Z^d)$-potentials, Rozenblum and Solomyak \cite{RoS09} gave an upper bound on the number of discrete eigenvalues in terms of the $\ell^{d/2}(\Z^d)$-norm of the potential.

 Recently, Ito and Jensen \cite{IJ16} expanded the free resolvent integral kernel in momentum space near each threshold energy, i.e., near each $\lambda\in  \tau(\Delta)$. This analysis complements that of \cite{AIM16}, which also relies on a detailed investigation of the resolvent in its momentum representation. 

For closely related problems, we mention that the results of Sahbani and Boutet de Monvel \cite{BS99}, were recently extended to general lattices by Parra and Richard \cite{PR16}. In addition,
the result in \cite{IM14}, on absence of embedded eigenvalues (away from $\{-d,d\}$) for finitely supported potentials, has a generalization to other lattices in \cite[Thm.~5.10]{AIM16}, where the inverse problem is also considered.
Finally, scattering on periodic metric graphs associated with $\Z^d$ was considered by Korotyaev and Saburova in \cite{KS15}.

%\medskip

\subsection{Estimates of free time evolution}

%Estimates for the free time evolution

Our first result consists of weighted estimates for the propagator
of the discrete Laplacian on $\Z^d$. To fix some
notation, we write $\ell^p(\Z^{d})$ for the space of sequences
$f=(f_n)_{n\in \Z^d}$ equipped with the norm
\[
\|f\|_{p}^p=\|f\|_{\ell^p(\Z^{d})}^p=\sum_{n\in \Z^d}|f_n|^p<\iy.
\]
For $p\geq 1$ and $\kappa\geq 0$, we shall make use of the weighted spaces
$\ell^p_\kappa(\Z^d)$, consisting of sequences with finite weighted norm
\[
\|f\|_{p,\kappa}^p = \|f\|_{\ell^p_\kappa(\Z^d)}^p= \sum_{n\in\Z^d} \rho_n^{-p\kappa} |f_n|^p,
\]
where
\[\label{rho-weight}
\r_n=\prod_{j=1}^d(1+|n_j|)^{-1},\qqq n\in \Z^d.
\]
For $p=\infty$ this amounts to $\|f\|_{\infty,\kappa} = \sup_n \rho^{-\kappa}_n|f_n|$.
For bounded operators $T$ on $\ell^2(\Z^d)$, we write $\|T\|$ for the operator norm of $T$.

Now we are now ready to formulate our first result on weighted estimates on the propagator $e^{it\D}$.

\begin{theorem}\label{T1} Let $d\geq 1$. The following holds true:
\lb{Texp}
\begin{enumerate}[label=\textup{(\alph*)}]
\item\label{Item-Texp-a} Let $q\geq 2$.
% and set $p:= \frac{2p_1p_2}{p_1+p_2}\geq 1$. If
Then for all $t\in \R\sm [-1,1]$ and $u,v\in\ell^{q}(\Z^d)$, we have
\begin{equation}
\label{et1} \|u e^{it\D}v\|\le
C_\Lan^\frac{2d}{q}|t|^{-{2d\/3q}}\|u\|_q\|v\|_q,
\end{equation}
where the constant $C_\Lan <  4/5$.
\item\label{Item-Texp-b} Let $a>1/2$, $2\leq q \leq\infty$ and $0\leq \kappa\leq a\cdot\frac{q-2}{q}$. Then for all $t\in\R$ with $|t|\geq 1$ and $u,v\in \ell^q_\kappa(\Z^d)$, we have
\begin{equation}
\label{et2}
\|u e^{it\D}v\|\le
C_\Lan^{\frac{2d}{q}} C_a^{\frac{d\kappa}{a}}
|t|^{-d\big({2\/3q}+{\kappa\/2a}  \big)}  \|u\|_{q,\kappa}  \|v\|_{q,\kappa},
\end{equation}
where  $C_a=3 \bigl(1+\frac{2a}{2a-1}\bigr)$. (For $q=\infty$ the obvious interpretation of the estimate applies.)
\end{enumerate}
\end{theorem}

\begin{remarks}
\begin{enumerate}[label = \textup{(\arabic*)}]
\item One can find Landau's optimal constant $C_\Lan$
with several decimals in \cite{La00}.
\item By allowing for the $\rho$-weights in \ref{Item-Texp-b}, one may improve the time-decay of the $\ell^q$-weighted time evolution
from \ref{Item-Texp-a} slightly.
\end{enumerate}
\end{remarks}

\subsection{Estimates of the free resolvent}

Let $\gamma\in [0,1]$, $d>2+2\gamma$ and $q\geq 2$. If $q>2$, we define
\begin{equation}
\lb{dGCd}
\G(q,d,\g)=\ca \Bigl(3+\frac{12(q-2)}{12-(5+2\gamma)q}\Bigr)^{\frac{3(q-2)}{q}}, & \textup{if} \ d =  3\\
 \Bigl(3+ 2\bigl(\frac{5q-2}{8-(3+\gamma)q} \bigr)^{1+\frac{q}{4(q-2)}}\Bigr)^{\frac{4(q-2)}{q}} & \textup{if} \ d=4 \\
 \Bigl(3+ \frac{6d(q-2)}{6d-(2d+1+3\gamma)q} \Bigr)^\frac{d(q-2)}{q} & \textup{if} \ d>4
 \ac.
\end{equation}
If $q=2$ and $d\neq 4$, we set $\Gamma(2,d,\gamma) = 1$ and
finally we set $\Gamma(2,4,\gamma) = (1-\gamma)^{-1}$.
We furthermore need
\[
 C_d^\gamma =\ca
\frac8{1-2\gamma}+8 & \textup{if} \ d=3,\\
\frac4{1-\gamma} & \textup{if} \ d=4,\\
\frac{14\cdot 2^\frac{d}{4}}{d-4} & \textup{if} \ d>4.
\ac
\]
and
\[\label{gamma-def}
\gamma_{d,q} = \ca \frac6{q} - \frac{5}{2}, & d=3 \\
\frac{2d}{q}- \frac{2d+1}{3}, & d\geq 4.\ac
\]

%\[
%\f_d=\ca {1\/6}, & d=3\\
%{1\/4}, & d=4\\
%{d-1\/3d}, & d\ge 5              \ac
%\]

Our second main theorem is the following weighted resolvent estimates.

\begin{theorem}
\lb{T2} Let $d\geq 3$. Let $u,v\in \ell^q(\Z^d)$ with
\[\lb{Vc}
2\leq q < \ca {12\/5} & \textup{if} \ \ d=3,\\
  \frac{6d}{2d+1} & \textup{if}  \  \ d\ge 4. \ac
\]
 Then the
operator-valued function $Y_0\colon \C\sm [-d,d]\to \cB(\ell^2(\Z^d))$, defined by
\[
Y_0(z) := u (\D -z)^{-1}v
\]
is analytic and H\"older continuous up to the boundary. More precisely, it satisfies:
\begin{enumerate}[label = \textup{(\alph*)}]
\item For all $z\in \C\sm [-d,d]$, we have
\begin{equation}
\label{VRV1}
\|Y_0(z)\|\le (1 +C_d^0 \G(q,d,0)) \|u\|_q\|v\|_q.
\end{equation}
\item Let $\gamma\in [0,1]$ satisfy the constraint $\gamma < \gamma_{d,p}$. For all $z, z'\in \C\sm [-d,d]$ with $\Im(z), \Im(z')\geq 0$, we
have
\begin{equation}\label{VRV2}
\|Y_0(z)-Y_0(z')\|\le |z-z'|^\g\bigl(1 +C_d^\gamma \G(q,d,\g)
\bigr)\|u\|_q\|v\|_q.
\end{equation}
\end{enumerate}
\end{theorem}

\begin{remark}\label{RemLip} We observe the following consequences:
\begin{enumerate}[label = \textup{(\arabic*)}]
\item The weighted resolvent $Y_0(z)$ is analytic in $\C\setminus [-d,d]$ and extends by continuity from $\C_\pm$ to
a H\"older continuous function on $\overline{\C}_\pm$. We denote the extension from $\C_\pm$ to $[-d,d]$ by $Y_0(\lambda\pm i0)$.
\item\label{Item-RemLip-b} If $d=3$ and $q=2$ the H\"older exponent must satisfy the constraint $\gamma<1/2$ and for $d=4$ and $q=2$ we have $\gamma<1$. However, if $d\geq 5$ and $q=2$, the extended weighted resolvent is Lipschitz continuous. In fact, one retains Lipschitz continuity for $d\geq 5$ provided
\[
2\leq q <  \frac{6d}{2d+4}.
\]
\item At the cost of a possibly larger prefactor, we actually get resolvent estimates in Hilbert-Schmidt norm. See Theorem~\ref{T2HS} for a precise formulation.
\end{enumerate}
\end{remark}

 In \cite[Lemmas~5.3 and 5.4]{IK12}, the authors establish free resolvent estimates in $d=1$ and  $d=2$ with $\ell^2(\Z^d)$-weights. These estimates blow up at the threshold set $\tau(\Delta)$, which is no coincidence, cf. \cite{IJ16}. In dimensions $d\geq 3$, \cite[Lemma 5.5]{IK12} establishes free resolvent estimates across thresholds but with mixed $\ell^2$-weights and $\rho$-weights. 
 In \cite{IK12} a H\"older estimate on a suitably weighted free resolvent is also derived, but with no control over the H\"older exponent $\gamma$.
 
% Estimates on weighted resolvents $F_s(z)= \rho^s(\Delta-z)^{-1}\rho^s$, where the weight $\rho$ was %defined in \eqref{rho-weight}, has also been considered in the literature. In $d=2$, bounds up to %non-threshold parts of the spectrum has been show to hold in $d=2$ with $s>1/2$, cf. 
%The limiting absorption principle has been established for $F_s(z)$ away from the threshold set %$\tau(\Delta)$ in different combinations of dimension $d$ and exponent $s$.   
%DISCUSS PAPERS.

\subsection{Applications to Schr\"odinger operators}
Finally we investigate some consequences of the resolvent estimates
for a particle in a cubic lattice, subject to an external potential.
We consider the Schr\"odinger operator $H=\D+V$ acting in
$\ell^2(\Z^{d}),d\ge 3$, where the real-valued potential $V = (V_n)_{n\in\Z^d}$
is an element of a suitable $\ell^p(\Z^d)$ space.
In particular, $H$ is a bounded operator with
$\sigma_\mathrm{ess}(H) = [-d,d]$.

Recall the decomposition 
\[\label{SpecDecomp}
\ell^2(\Z^d) = \cH_\mathrm{ac} \oplus \cH_\mathrm{sc}\oplus \cH_\mathrm{pp}
\]
of the Hilbert space into the absolutely continuous, singular continuous subspaces of $H$ and the closure of the span of all eigenstates of $H$.
 
\begin{theorem}\lb{T3}
Let $d\geq 3$ and $V\in\ell^p(\Z^d)$ with
\[
1\leq p <
\ca {6\/5} & \textup{if} \ \ d=3,\\
  \frac{3d}{2d+1} & \textup{if}  \  \ d\ge 4. \ac
\]
\begin{enumerate}[label = \textup{(\alph*)}]
\item\label{Item-T3-a} All eigenvalues $\lambda\in\sigma_\mathrm{pp}(H)$ are of finite multiplicity.
\item\label{Item-T3-a1} The closure of the set $\sigma_\mathrm{pp}(H)\cup \sigma_\mathrm{sc}(H)$ has zero Lebesgue measure.
If, in addition, $\|V\|_p < (1+ C_d^0 \Gamma(2p,d,0))^{-1}$, then:
$\sigma_\mathrm{pp}(H) = \sigma_\mathrm{sc}(H)= \emptyset$.
\item\label{Item-T3-a2}
 The wave operators 
\[
W^\pm: = \slim_{t\to\pm\infty} e^{it H}e^{-it\Delta}
\]
exist and are complete, i.e.; $W^\pm\ell^2(\Z^d)=\cH_\mathrm{ac}$.
\item\label{Item-T3-b} If $d\geq 5$ and
 \[\label{Crit-p}
 1\leq p <
   \frac{3d}{2d+4},
 \]
then $\sigma_\mathrm{pp}(H)$ is a finite set and $\sigma_\mathrm{sc}(H) = \emptyset$.
\end{enumerate}
\end{theorem}

\begin{remark} Since $p\leq d/2$ and $d\geq 3$, it is a consequence of a result of Rozenblum and Solomyak \cite[Thm.~1.2]{RoS09} that $\sigma(H)\setminus [-d,d]$ is a finite set. Finiteness of $\sigma_\mathrm{pp}(H)\cap [-d,d]$ is typically a much harder question, here settled in the affirmative for $d\geq 5$. 
\end{remark}

\begin{example} Let $d \geq 3$ and
let $\kappa\colon \Z^d\to\Z^d$ be injective and $\theta\colon\Z^d\to \C$ satisfy $|\theta_n|\leq 1$ for all $n$.
Define a potential by setting
\[
V_n = \ca 0 & \textup{if} \ n\not\in \kappa(\Z^d) \\
\theta_n \prod_{j=1}^d\rho_{\kappa^{-1}(n)_j} & \textup{if} \ n\in\kappa(\Z^d)\ac
\]
Consider the Hamiltonian
\[
H_g = \Delta+ g V.
\]
We conclude from Theorem~\ref{T3} that:
\begin{enumerate}[label=\textup{(\arabic*)}]
\item\label{Item-Rem-T3-1} the wave operators exist and are complete. If in addition $g$ is sufficiently small, more precisely $|g|< (1+C_d^0 \Gamma(11/10,d,0))^{-1}\|\rho\|_{11/10}^{-d}$, then the operator
$H_g$ has no eigenvalues and the singular continuous spectrum is empty.
\item if $d\geq 5$, then we additionally have, for any $g\in\R$, that the operator $H_g$
has at most finitely many  eigenvalues, all of finite multiplicity, and the singular continuous spectrum is empty.
\end{enumerate}
Note that the critical coupling in \ref{Item-Rem-T3-1} does not depend on $\kappa$ and $\theta$. Even if $\kappa(n) = n$ and $\theta_n=1$ all $n$, the potential is long-range in the direction of each coordinate axis. Playing with $\kappa$, one can engineer sparse potentials with arbitrarily slow decay.
\end{example}

We end this section with a discussion of the momentum representation of the discrete Laplacian.

One may diagonalize the discrete Laplacian, using the (unitary) Fourier transform
 $\F\colon \ell^2(\Z^d)\to L^2(\T^d)$, where $\T=\R/(2\pi \Z)$. It is defined by
 \[
 (\F f)(k)=\wh f(k)={1\/(2\pi)^{{d\/2}}}\sum_{n\in
 \Z^d} f_ne^{i n\cdot k},\qq \textup{where} \qq
 k=(k_j)_{j=1}^d\in \T^d.
 \]
 Here $k\cdot n = \sum_{j=1}^d k_j n_j$ is the scalar product in $\R^d$.
In the resulting momentum representation of the discrete Laplacian $\Delta$, we write $\wh \Delta =\F \Delta \F^*$.
We recall that the Laplacian is transformed into a multiplication operator
\[
(\wh \D \wh f)(k)=\Bigl(\sum_{j=1}^d \cos k_j\Bigr) \wh f. %=\sum_{j=1}^d
%\rt(-1+2\cos^2 {k_j\/2}\rt).
\]

 The operator $e^{it \D}, t\in \R$ is unitary  on $L^2(\T^d)$ and has  the kernel
$(e^{it \D})(n-n')$, where for $n\in\Z^d$:
\[
\begin{aligned}
\lb{DiscAsBessel}
(e^{it \D})(n) & ={1\/(2\pi)^{d}}\int_{\T^d}
e^{-i n\cdot k+it \sum_{j=1}^d \cos(k_j)}dk\\
& = \prod_{j=1}^d \Bigl(\frac1{2\pi}\int_0^{2\pi}e^{-in_jk+it\cos(k)}dk\Bigr)
= i^{|n|} \prod_{j=1}^d J_{n_j}(t),
\end{aligned}
\]
where $|n| = n_1+\cdots + n_d$.
Here $J_n(z)$ denotes the Bessel function:
\[\label{BesselFunction}
J_n(t)={(-i)^n\/2\pi}\int_0^{2\pi} e^{in k-i t\cos(k)}dk \qqq \forall \
(n,z)\in \Z \ts \R.
\]
We have collected some basic properties of Bessel function with integer index in \eqref{Be}.

\subsection{Plan of the paper}
In Section~\ref{SecTimeEvolution} we determine properties of the free time evolution and prove Theorem~\ref{T1}.  Section~\ref{SecResolvent} contains basic estimates on the free
resolvent and a proof of Theorem~\ref{T2}. In Section~\ref{SecHamiltonian} we apply the above results to discrete Schr\"odinger
operators and prove Theorem~\ref{T3}. In Appendix~\ref{AppA} and~\ref{AppB} we recall and expand on some key estimates of Bessel
functions. Finally in Appendix~\ref{AppC}, we have collected some useful discrete estimates.

\section{Estimates for the free time evolution}\label{SecTimeEvolution}
\setcounter{equation}{0}

Our first lemma establishes a basic mapping property of the summation kernel of the
free propagator $e^{i t \Delta}$.

\begin{lemma}\lb{Texp1}
 Let $2 \le r\le \iy, s\in [1,2]$ and ${1\/r}+{1\/s}=1$.
Then for all $t\in\R$
 \[
 \lb{Eet}
 \|e^{it\D}f\|_r\le  C_\Lan^{{2d}({1\/s}-{1\/2})}|t|^{-{2d\/3}({1\/s}-{1\/2})}  \|f\|_{s}\qqq \forall \
 f\in \ell^s(\Z^d).,
 \]
 where $C_\Lan$ is one of Landaus optimal constants from Lemma~\ref{Landau}.
\end{lemma}

\begin{proof}
Note that the operator $e^{it\D}$ is unitary on $\ell^2(\Z^d)$. In particular $\|e^{it\D}f\|_2=\|f\|_2$.

If $f \in \ell^1(\Z^d)\cap \ell^2(\Z^d)$, then it follows from \er{LandauC} that
$\|e^{it\D}f\|_\iy\le C_\Lan^d t^{-{d\/3}}\|f\|_1$. By the discrete Riesz-Thorin Interpolation Theorem (Theorem~\ref{TRT}),  $e^{it\D}$ extends uniquely to a map from $\ell^{s}(\Z^d)$ to $\ell^r(\Z^d)$
and satisfies \er{Eet}.
\end{proof}

%The function
%$$
%(e^{it\D})(n)=i^{|n|}\prod_{j=1}^d J_{n_j}(t),\qqq (n,t)\in %\Z^d\ts \R,
%$$
%is called the free propagator on the lattice.

Now we describe the more regular case.

\begin{lemma}
\lb{Td1}
 Let  $a>{1\/2}, c\in [0,1]$ and let $c\in \R$ with $|t|\geq 1$.
Abbreviate $C_a =3\bigl(1+\frac{2a}{2a-1}\bigr)$.
\begin{enumerate}[label=\textup{(\alph*)}]
\item\label{Item-Td1-a}  If $d=1$, then the following
estimates hold true:
\begin{equation}
\lb{d1x}  \|\r^ae^{it\D}\r^a\|\le C_a |t|^{-\frac12},
\end{equation}
\begin{equation}
\lb{d1xx}  \|\r^{ac}e^{it\D}\r^{ ac}\|\le C_a^c |t|^{-{c\/2}}.
\end{equation}
\item\label{Item-Td1-b} If $d\ge 1$, then the following estimates hold true:
\[
\lb{dx}  \|\r^ae^{it\D}\r^a\|\le C_a^d|t|^{-{d\/2}},
\]
\[
\lb{dxx}  \|\r^{ ac}e^{it\D}\r^{ac}\|\le C_a^{dc}|t|^{-{c d\/2}}.
\]
\end{enumerate}
\end{lemma}

\begin{proof}
To establish \ref{Item-Td1-a}, we estimate for $d=1$ using Proposition~\ref{FusedBessel}
\[
\begin{aligned}
&\rho(n)^{2a}|e^{it\Delta}(n-m)|^2\rho(m)^{2a} = \frac{|J_{n-m}(t)|^2}{(1+|n|)^{2a}(1+|m|)^{2a}}\\
&\qquad \leq t^{-\frac12} \frac{1}{(1+|n|)^{2a}(1+|m|)^{2a}(|n-m|^{\frac13}+||n-m|-t|)^{\frac12}},
\end{aligned}
\]
for all $t>0$ and $n,m\in\Z$.
Invoking \eqref{Landau0} and Lemma~\ref{SummationEstimate} (with $\alpha=2a$ and $\beta = 1/2$), we estimate in Hilbert-Schmidt norm:
\[
\begin{aligned}
& \|\rho^a e^{-it\Delta}\rho^a\|_{\cB_2}^2 \leq \sum_{n,m\in\Z} \rho(n)^{2a}|e^{it\Delta}(n-m)|^2\rho(m)^{2a} \\
&\qquad \leq \frac{2}{t\pi} \sum_{n'\in \Z} \frac1{(1+|n'|)^{4a}}
+ t^{-\frac12}\sum_{n,m\in\Z} \frac{1}{(1+|n|)^{2a}(1+|m|)^{2a}(1+||n-m|-t|)^{\frac12}}\\
&\qquad \leq
 \frac{2}{t\pi}\Bigl(1+2 \int_0^\infty (1+x)^{-4a}\, dx\Bigr)+\frac{1}{t}
 \Bigl(\frac{2\alpha^2}{(\alpha-1)^2}+
 \frac{4\alpha}{(\alpha-1)(p\alpha-1)^{\frac1{p}}} 32^\frac1{r} \Bigr)\\
 &\qquad\leq
  \frac{6}{t\pi}+\frac{1}{t}
  \Bigl(\frac{2 (2a)^2}{(2a-1)^2}+
  \frac{4(2 a)}{(2a-1)(6 a-1)^{\frac1{3}}} 32^\frac2{3} \Bigr)\\
  &\qquad
  \leq \frac2{t} + \frac{1}{t}
    \Bigl(\frac{2 (2a)^2}{(2a-1)^2}+
    \frac{4( 2a)}{(2a-1)2^\frac13 } 32^\frac2{3} \Bigr)\\
&\qquad  = \Bigl(2\frac{(2a)^2}{(2a-1)^2} + 32 \frac{2a}{2a-1}+ 2\Bigr)\frac1{t}\leq 9 \Bigl(1+\frac{2a}{2a-1}\Bigr)^2.
 % \frac{2(2a)^2}{(2a-1)^2}\Bigl(1+ 24\cdot 4^{2a} \Bigr).
\end{aligned}
\]
This proves \eqref{d1x} and \eqref{d1xx} now follows from the estimate
$\|T\|\leq \|T\|_{\cB_2}$ valid for Hilbert-Schmidt operators, and Hadamard's Three Line Theorem
applied for $\psi\in\ell^2(\Z^d)$ with $\varphi(z) = \langle \psi, \rho^{az} e^{it\Delta}\rho^{az}\psi\rangle$ for $z\in\C$ with $0\leq \Re z \leq 1$. See \cite[App. to IX.4]{RS75}.

To conclude, we observe that the statements in \ref{Item-Td1-b} follow from \ref{Item-Td1-a}.
\end{proof}

 \begin{proof}[Proof of Theorem \ref{Texp}]
We begin with \ref{Item-Texp-a}. It suffices to proof the claim for $t\geq 1$.
 From \eqref{Eet}, we recall the estimate
 \[\label{thm1.1-step1}
\forall \ f\in \ell^s(\Z^d):\qquad \|e^{it\D}f\|_r\leq C_\Lan^\kappa t^{-\frac{\kappa}{3}}   \|f\|_{s},\qqq \textup{where} \ \kappa=2d\Bigl({1\/s}-{1\/2}\Bigr),
 \]
and $2\le r\le \iy$ with $ {1\/r}+{1\/s}=1$. Thus, if $u\in \ell^q(\Z^d)$ with $q\ge 2$, then $r = (1/2-1/q)^{-1}\geq 2$ such that we may estimate for $u,v\in\ell^q(\Z^d)$ using \eqref{thm1.1-step1} and H\"older's inequality
\[
 \|ue^{it\D}v \vp\|_2\le   \|u\|_q \|e^{i t \D} v\vp\|_r \leq  C_\Lan^{\frac{2d}{q}}t^{-{2d\/3q}} \|u\|_{q} \|v\vp\|_s \leq C_\Lan^{\frac{2d}{q}}t^{-{2d\/3q}} \|u\|_{q} \|v\|_{q}  \|\vp\|_2,
\]
which yields \er{et1}.

We now turn to \ref{Item-Texp-b}. It follows from \er{et1} and \er{dxx} that we have
\[\lb{ixx}
\begin{aligned}
\forall a>0,\,c\in[0,1]:\qquad & \|\r^{ac}e^{it\D}\r^{ac}\|\le C_a^{cd}|t|^{-{dc\/2}},\\
\forall q'\geq 2:\qquad &\|ue^{it\D}v\|\le C_\Lan^{\frac{2d}{q'}}|t|^{-{2d\/3q'}}\|u\|_{q'}\|v\|_{q'}.
\end{aligned}
\]
Here $C_a$ is the explicit constant from Lemma~\ref{Td1}.

Let $\psi\in\ell^2(\Z^d)$, $u,v\colon\Z^d\to [0,\infty)$ with finite support,
and put
\[
\varphi(z) = \langle \psi, \rho^{ac(1-z)} u^z e^{it\D}v^z \rho^{ac(1-z)} \psi\rangle.
\]
 Writing $\varphi$ as a double sum, we see that $\varphi$ is analytic in $0<\Re z <1$ and continuous on the closure of the strip. Moreover, $\varphi$ is bounded in the closed strip so we may apply Hadamard's Three Line Theorem.
The estimates in \er{ixx} now interpolates to the estimate
\[
\lb{ixxx}
\|\r^{abc}u^{1-b}e^{it\D}v^{1-b}\r^{abc}\|\le C_\Lan^{\frac{2 d(1-b)}{q'}} C_a^{dbc}|t|^{-\b}\|u\|_{q'}^{1-b}\|v\|_{q'}^{1-b},
\qqq
\]
valid for $u,v\colon \Z^d \to [0,\infty)$.
Here
\[
\b={2d\/3q'}(1-b)+{d\/2}b=d\rt({2(1-b)\/3q'}+{bc\/2}  \rt).
\]
For $u,v\colon \Z^d\to \C$ with finite support, we can now estimate
for $b,c\in[0,1]$ and $q'\geq 2$:
\[
\begin{aligned}
\|u e^{it\D}v\| &= \||u| e^{it\D}|v|\| = \|\rho^{abc}(u')^{1-b}e^{it\D} (v')^{1-b}\rho^{abc}\|\\
&\leq C_\Lan^{\frac{2 d(1-b)}{q'}} C_a^{dbc}
 |t|^{-\beta} \|u'\|_{q'}^{1-b}\|v'\|_{q'}^{1-b},
\end{aligned}
\]
where $u' = \rho^{-\frac{abc}{1-b}}|u|^{\frac1{1-b}}$ and  $v' = \rho^{-\frac{abc}{1-b}}|v|^{\frac1{1-b}}$.

Let $q\geq 2$ and $0\leq \kappa\leq a(q-2)/q$ be the exponents from the formulation of the theorem.
In order to end up with a $q$ norm, we must pick $b\in[0,1]$ such that $q' = q(1-b)$ and to ensure $q'\geq 2$, we get the constraint $0\leq b \leq (q-2)/q$. We get best time decay by picking $b = (q-2)/q$, in which case we end up with the estimate
\[
\|u e^{it\D}v\|\leq C_\Lan^{\frac{2d}{q}} C_a^{\frac{dc(q-2)}{q}} t^{-d\bigl(\frac{2}{3q} + \frac{c(q-2)}{2q}\bigr)} \|u\|_{q,ac\frac{q-2}{q}}
\|v\|_{q,ac\frac{q-2}{q}}.
\]
 If $q>2$, then
we may choose
$c = a\kappa q/(q-2) \in[0,1]$ to arrive at the desired estimate \er{et2},
at least for $u,v$ with finite support. If $q=2$, the same conclusion holds immediately. After extension by continuity to arbitrary $u,v\in \ell^q_\kappa(\Z^d)$, we conclude the proof.
\end{proof}

\section {Estimates on the free resolvent}\label{SecResolvent}
\setcounter{equation}{0}

We have the standard representation of the free resolvent $R_0(\l)$
in the lower  half-plane $\C_-$ given by
\[
\lb{R012}
\begin{aligned}
R_0(\l)=-i\int_0^\iy e^{it(\D-\l)}dt=R_{01}(\l)+R_{02}(\l),
\\
R_{01}(\l)=-i\int_0^1 e^{it(\D-\l)}dt,\qqq R_{02}(\l)=-i\int_1^\iy
e^{it(\D-\l)}dt,
\end{aligned}
\]
for all $\l\in \C_-$.
 Here the operator valued-function
$R_{01}(\l)$ has analytic extension from  $\C_-$ into the whole
complex plane $\C$ and
 satisfies
\[
\lb{Rest1}
\forall\l,\m\in\overline{\C}_-:\qquad \|R_{01}(\l)\|\le 1\qquad \textup{and} \qquad \|R_{01}(\l)-R_{01}(\m)\|\le |\l-\m|.
\]

We shall need the following lemma.

\begin{lemma}\lb{TqEp}
 Let $d\ge 3$ and $\gamma\in [0,1]$ be such that $d>2+2\gamma$. Then for each $(t,n)\in \R\ts \Z^d$ the following
estimate holds true:
\[
\lb{Int} \int_{1}^\infty t^\gamma \bigl|(e^{it\D})(n)\bigr|\, dt\le
C_d^\gamma\wt\vk_d^\gamma(n), \qq
\]
where the constant $C_d^\gamma$ is defined by \er{Cpgamma},
\[\label{tildekappa}
\wt\kappa_d^\gamma(n) = \prod_{j=1}^d \kappa_d^\gamma(n_j)
\]
 and $\vk^\gamma_d$ is defined in \er{kpgamma}.
\end{lemma}

\begin{proof}
Using \er{DiscAsBessel} and \er{es1}, we obtain for all $n\in\Z^d$:
\[
\int_{1}^\infty t^\g|(e^{it\D})(n)|dt=\int_{1}^\infty \prod_{j=1}^d
t^{\g\/d}|J_{n_j}(t)|dt\le \prod_{j=1}^d \rt (\int_{1}^\infty t^\g
|J_{n_j}(t)|^d\rt)^{1/d}\le  C_d^\gamma \wt\vk_d^\gamma(n),
\]
which yields \er{Int}.
 \end{proof}

The above lemma yield estimates on the summation kernel of $R_{02}(\lambda)$.

\begin{proposition}\lb{PR0} Let $d\geq 3$.
The following estimates for the summation kernel of $R_{02}(\lambda)$ hold true:
\begin{enumerate}[label= \textup{(\alph*)}]
\item For all $m\in\Z^d$ and $\l\in\C\setminus\R$:
\[ \lb{R02}
|R_{02}(m,\l)|\le C_d^0\wt\vk_d^0(m),
\]
where $\wt \vk_d^0$ is defined in \eqref{tildekappa}.
\item  Suppose $\g\in [0,1]$
and $d> 2+2\gamma$. Then for all $m\in\Z^d$ and $\l,\m\in\C\setminus \R$:
\[\lb{R02x}
|R_{02}(m,\l)-R_{02}(m,\m)|\le C_d^\gamma \wt\vk_d^\gamma(m)
|\l-\m|^\g.
\]
\end{enumerate}
\end{proposition}

\begin{proof} Consider first the case $\l,\m\in\C_-$.
% We have for $m\in \Z^d$ the
%identity
%\[
%\lb{Ret} R_{02}(m,\l)=-i\int_1^\iy (e^{it(\D-\l)})(m)\,dt.
%\]
From \eqref{R012} and the estimate \er{Int} (with $\gamma=0$), we find that
\[
\lb{ER1} |R_{02}(m,\l)|\le \int_1^\iy|(e^{it\D})(m)|dt \le
C_d^0\wt\vk_d^0(m).
\]
where the sequence $\wt\vk_p^0=(\wt\vk_p^0(n))_{n\in \Z^d}$ is given by \er{tildekappa}.

The H\"older estimate \er{R02x}, follows from the identity \er{R012} and the estimates $|e^{-i t\l}-e^{-it\m}|\leq t^\gamma|\l-\m|^\gamma$ and \er{Int}:
$$
|R_{02}(m,\l)-R_{02}(m,\m)|\le |\l-\m|^\g \int_1^\iy t^\g
|(e^{it\D})(m)|dt\le C_d^\gamma \wt\vk^\g_d(m)|\l-\m|^\g.
$$

For $\lambda,\mu\in\C_+$, the two estimates follow from what has already been proven together with
the identity $R_{02}(m,\overline{z}) = \overline{R_{02}(-m,z)}$.
\end{proof}

\begin{proof}[Proof of Theorem \ref{T2}]
%({\bf  Under the reconstruction})
 Consider the case $\l\in \C_-$ and let $u,v\in\ell^q(\Z^d)$ with $d\geq 3$. Define $Y_{02}$ by
 $$
Y_{02}(z)= u\, R_{02}(z)\, v,\qqq \ \forall z\in\C\sm
[-d,d].
 $$
The estimate \er{R02} yields an estimate of the Hilbert-Schmidt norm of $Y_{02}(z)$:
$$
\begin{aligned}
\|Y_{02}(z)\|_{\cB_2}^2 & =\sum_{n, n'\in \Z^d}|u_n|^2\
|R_{02}(
n'-n,z)|^2|v_{n'}|^2\\
& \le (C_d^0)^2 \sum_{n, n'\in \Z^d}|u_n|^2\
\wt\vk_d^0(n'-n)^2|v_{n'}|^2
\end{aligned}
$$
and applying the Young inequality from Theorem \ref{YE} at
${2\/q}+{1\/r}=1$,  we obtain
\[
\begin{aligned}
\|Y_{02}(z)\|_{\cB_2}^2& \le (C_d^0)^2 \sum_{n,n'}\wt\vk_d^0(n'-n)^2|u_n|^2 |v_{n'}|^2\\
& \le (C_d^0)^2\|u\|_{q}^2 \|v\|_q^2\|\wt\vk_d^0\|_{r}^2.
\end{aligned}
\]
Note that
\[
\lb{vk0}
\begin{aligned}
\|\wt\vk_d^0\|_{r}= \rt(\sum_{n\in
\Z^d}\wt\vk_d^0(n)^r\rt)^{1\/r}=\rt(\sum_{n\in \Z^d} \prod_{j=1}^d
\vk_d^0(n_j)^r\rt)^{1\/r}=\|\vk_d^0\|_{r}^{d}.
 \end{aligned}
\]
In order to ensure finiteness of $\|\wt\vk_d^0\|_{r}$
we must therefore, according to Lemma~\ref{lemma-kappa-r-bound}, require that
$r>r_d^0$, where $r_d^0$ is defined in \eqref{r-p-gamma}. Note $q = 2r/(r-1)$ and that this expression is decreasing in $r\geq 2$. Since $r_d^0\geq 2$, we may therefore
express the constraint $r>r_d^0$ by the constraint
$q < 2 r_d^0/(r_d^0-1)$ on $q$ instead. Inserting the expression for $r_d^0$ from \eqref{r-p-gamma}, we arrive at the constraint in \eqref{Vc}. Furthermore, employing the first estimate in \eqref{Rest1}, and the estimate from Lemma~\ref{lemma-kappa-r-bound} with $\gamma=0$, we arrive at \eqref{VRV1}. (Recall that $\|T\|_{\cB_2}\leq \|T\|$.)

%Here we used the restriction imposed in \eqref{Vc} on $q$, which %ensures that $r>r_d^0$, as required in %Lemma~\ref{lemma-kappa-r-bound}.

The estimate \er{R02x} and  the Young inequality from Theorem
\ref{YE} with ${1\/p}+{1\/r}=1$,  implies
$$
\begin{aligned}
\|Y_{02}(z)-Y_{02}(z')\|_{\cB_2}^2 &=\sum_{n, n'\in \Z^d}|u_n|^2\
|R_{02}(n'-n,z)-R_{02}(n'-n,z')|^2|v_{n'}|^2\\
&\le (C_d^\gamma)^2|\l-\m|^{2\g} \sum_{n, n'\in \Z^d}|u_n|^2\
\wt\vk_d^\g(n'-n)^2|v_{n'}|^2\\
&\le (C_d^\gamma)^2|\l-\m|^{2\g}\|u\|_{q}^2\|v\|_q^2\|\wt\vk_d^\g\|_{r}^2
\le
(C_d^\gamma)^2|\l-\m|^{2\g}\|u\|_{q}^2\|v\|_q^2\G^2(q/2,d,\g),
\end{aligned}
$$
 since -- due to Lemma~\ref{lemma-kappa-r-bound} -- we have:
\[
\lb{vk1}
\|\wt\vk_d^\g\|_{r}= \rt(\sum_{n\in
\Z^d}\wt\vk_d^\gamma(n)^r\rt)^{1\/r}=\rt(\sum_{n\in \Z^d} \prod_{j=1}^d
\vk_d^\gamma(n_j)^r\rt)^{1\/r}=\|\vk_d^\gamma\|_{r}^d\le \G(q,d,\g).
\]
This, together with the second estimate in \eqref{Rest1}, completes the proof.
\end{proof}

In the proof Theorem~\ref{T2}, we actually estimated the  $R_{02}(z)$
contribution in Hilbert-Schmidt norm, only the $R_{01}(z)$ contribution
was estimated directly in operator norm, cf. \eqref{Rest1}.

We end this section with resolvent estimates in Hilbert-Schmidt norm, where
we must investigate the $R_{01}$ contribution more closely.
Let $\cB_2$ denote the Hilbert-Schmidt class. We write
$\|K \|_{\cB_2}$ for the Hilbert-Schmidt norm of an operator $K$.

\begin{theorem}
\lb{T2HS} Let $d\geq 3$. Let $u,v\in \ell^q(\Z^d)$ with
\[\lb{HSVc}
2\leq q < \ca {12\/5} & \textup{if} \ \ d=3,\\
  \frac{6d}{2d+1} & \textup{if}  \  \ d\ge 4. \ac
\]
 Then the
operator-valued function $Y_0\colon \C\sm [-d,d]\to \cB_2$, defined by
\[
Y_0(z) := u (\D -z)^{-1}v
\]
is analytic and H\"older continuous up to the boundary. More precisely, it satisfies:
\begin{enumerate}[label = \textup{(\alph*)}]
\item For all $z\in \C\sm [-d,d]$, we have
\begin{equation}
\label{VRV1HS}
\|Y_0(z)\|_{\cB_2}\le \bigl(D_{q,d} +C_d^0 \G(q,d,0)\bigr) \|u\|_q\|v\|_q,
\end{equation}
where $D_{2,d}=1$ and $D_{q,d} = (q/2)^{\frac{d(q-2)}{q}}$
\item Let $\gamma\in [0,1]$ satisfy the constraint $\gamma < \gamma_{d,p}$. For all $z, z'\in \C\sm [-d,d]$ with $\Im(z), \Im(z')\geq 0$, we
have
\begin{equation}\label{VRV2HS}
\|Y_0(z)-Y_0(z')\|_{\cB_2}\le |z-z'|^\g\bigl(D_{q,d} +C_d^\gamma \G(q,d,\g)
\bigr)\|u\|_q\|v\|_q.
\end{equation}
\end{enumerate}
\end{theorem}

\begin{proof} We only need to estimate the $R_{01}$ contribution, since the $R_{02}$ contribution was estimated in  Hilbert-Schmidt norm in the proof of Theorem~\ref{T2}.

We abbreviate $Y_{01}(z) = u R_{01}(z) v$ and estimate using Lemma~\ref{lem-small-t}
\[
\|Y_{01}(z)\|_{\cB_2}^2 \leq  \sum_{n,m} |u_n|^2\Bigl(\int_0^1 |e^{it\Delta}(n-m)| \,dt\Bigr)^2|v_m|^2
\leq \sum_{n,m} |u_n|^2 \rho_{n-m}|v_m|^2.
\]
Recall from \eqref{rho-weight} the definition of the function  $\rho$.
Let $r=\infty$ if $q=2$ and $r  =(1-2/q)^{-1}= \frac{q}{q-2}$ if $q>2$.   
By the discrete Young inequality, Theorem~\ref{YE}, we conclude the estimate
\[
\|Y_{01}(z)\|_{\cB_2} \leq \|\rho\|_r\|\|u\|_q\|v\|_q.
\]
For $q>2$, we estimate (using \eqref{rho-alpha-norm})
\[
 \|\rho\|_r^r = \sum_{n\in\Z^d} \prod_{j=1}^d (1+|n_j|)^{-r}
 = \Bigl(\sum_{m\in\Z} (1+|m|)^{-r}\Bigr)^d\leq \Bigl( \frac{r}{r-1}\Bigr)^d
 = \frac{q^d}{2^d}.
\]

Similar, we estimate for any $\gamma\in [0,1]$ and $z,z'$ in the same half-plane
\[
\|Y_0(z)-Y_0(z')\|_{\cB_2}\leq |z-z'|^\gamma D_{q,d} \|u\|_q\|v\|_q.
\]
This completes the proof.
\end{proof}

\section {Schr\"odinger operators}\label{SecHamiltonian}
\setcounter{equation}{0}

%We recall the estimate from \cite{LS}:
%\[
%\dim(H<-d)+\dim(H>d)+\le C\sum_{n\in \Z^d} |V(n)|^{d\/2},
%\]
%for some absolute constant $C$ and $d\ge 3$.

Let $V\in L^p(\Z^d)$ with $p\geq 1$ and write
$q_1 = |V|^{1/2}$. Note that $q_1\in \ell^q(\Z^d)$ with $q = 2p\geq 2$.  
Choose $q_2\in \ell^q(\Z^d)$, such that $V = q_1 q_2$.
The specific choice $q_2 = q_1 \sign(V)$ would work, but we shall exploit the freedom
to choose $q_2$ differently in the proof of Theorem~\ref{LimAbs} below.
Note that for $n\in\supp(V) = \supp(q_1)$, we have $q_2(n) = \sign(V_n) q_1(n)$.

We say that $f\in L^2(\Z^d)$ solves the \emph{Birman-Schwinger equation} at $\lambda$ if
\begin{equation}\label{BirSch}
f = - q_1 R_0(\lambda+ i 0) q_2 f.
\end{equation}
We write
\[
\sigma_\BS(H) = \bigl\{ \lambda\in\R \,\big|\, \textup{The Birman-Schwinger equation has a non-zero solution at } \lambda\,\bigr\}.
\]
Note that any solution $f\in\ell^2(\Z^d)$ of the Birman-Schwinger equation satisfies
\[\label{SuppOfBSsol}
\supp(f) \subseteq \supp(V),
\]
which in particular implies that $\sigma_\BS(H)$ does not depend on the choice of $q_2$.

\begin{remark}\label{BSminus} In principle one should also consider the Birman-Schwinger equation with the limiting resolvent coming from the lower half-plane, defined by $q_1R_0(\lambda-i 0)q_2$. This would however give rise to the same Birman-Schwinger spectrum, so we do not introduce a separate notation. Indeed, choose
$q_2 = \sign(V) q_1$, and define a unitary transformation by setting $ (Uf)_n = \sign(V_n)f_n$ if $V_n\neq 0$ and $(Uf)_n = f_n$ otherwise. Then
\[
q_1 R_0(\lambda-i0)^{-1}q_2 = U^* q_2 R_0(\lambda -i0) q_1 U
= U^* \bigl(q_1 R_0(\lambda +i 0) q_2 \bigr)^* U,
\]
which implies that the Birman-Schwinger spectrum does not depend on the choice of limiting resolvent.
\end{remark}

Note that since $q_1 R_0(\lambda+ i 0) q_2$ is compact, the solution spaces to the Birman-Schwinger equation (for a given $\lambda$) is finite dimensional, i.e., all $\lambda\in \sigma_\BS(H)$ are of finite multiplicity.

\begin{lemma}\label{lem:TrivialLimit} For any $\l\in\R$, we have the limit $s-\lim_{\mu\to 0} \mu R_0(\lambda+i \mu)=0$.
\end{lemma}

\begin{proof} Let $g\in L^2(\Z^d)$ and $\epsilon>0$. Pick $\delta>0$ such that $\| \textbf{1} [|\D-\lambda|<\delta] g\| \leq \epsilon/2$. We may now, for $\mu\in\R$ with $0<|\mu|< \epsilon\delta/2$, estimate
$\|\mu(\D-\lambda+i \mu)^{-1} g\| \leq \epsilon$. This completes the proof.
\end{proof}

\begin{lemma}\label{lem:From-pp-to-BS} We have $\sigma_\pp(H)\subseteq \sigma_\BS(H)$ and the map
$g \to q_1 g$ takes nonzero eigenfunctions $Hg = \lambda g$ into nonzero solutions of \eqref{BirSch} at $\lambda$.
\end{lemma}

\begin{proof}
Suppose $H g = \lambda g$, for some non-zero $g\in L^2(\Z^d)$.
Put $f = q_1 g$  and compute
\[
q_1R_0(\lambda+ i \mu) q_2 f=  q_1R_0(\lambda+ i \mu) (\lambda- \D) g
= - q_1 g - i \mu q_1 R_0(\lambda+ i \mu)  g.
\]
 Taking the limit $\mu\to 0$, using Lemma~\ref{lem:TrivialLimit}, we arrive at $f = q_1 g$ being a solution of the Birman-Schwinger equation \eqref{BirSch}. It remains to argue that $f\neq 0$.
 If $f = 0$, then $V g = 0$ and consequently, $(\D-\lambda)g = (H-\lambda)g = 0$. This is absurd, since $\D$ does not have eigenvalues.
 \end{proof}

The Birman-Schwinger equation with limiting resolvent, has been used previously by Pushnitski \cite{Pu11} to study embedded eigenvalues, in view of the above lemma.

\vspace{5mm}

We are now in a position to give:

\begin{proof}[Proof of Theorem~\ref{T3}~\ref{Item-T3-a}]
 We show that eigenvalues of $H$ have finite multiplicity. Let $\lambda\in \sigma_\pp(H)$ and denote by $\cH_\lambda$ the associated eigenspace.
By Lemma~\ref{lem:From-pp-to-BS}, the linear map $\cH_\lambda\ni g \to q_1 g\in L^2(\Z^d)$ is injective and takes values in the vector space of solutions of \eqref{BirSch} at $\lambda$. Since this vector space is finite dimensional ($q_1 R_0(\lambda+ i 0) q_2$ being compact),
we may conclude that the eigenspace $\cH_\lambda$ is finite dimensional.
\end{proof}

In what remains of this section, we shall use the abbreviations
\[\label{YandwtY}
Y_0(z) = q_1 (\Delta-z)^{-1}q_2, \quad \textup{and} \quad
\wt Y_0(z) = q_2 (\Delta-z)^{-1} q_2
\]
for $z\in\C$ with $\Im z\neq 0$. 
 For $\lambda\in \R$, we write
$Y_0(\lambda\pm i0)$ and $\wt Y_0(\lambda\pm i 0)$ for the limiting objects. By a limiting argument, we observe that the identity $\wt Y_0(z) = \sign(V) Y_0(z)$ valid for $z\in\C$ with $\Im z\neq 0$ extends to
\[\label{tYtoY}
\wt Y_0(\lambda\pm i0) = \sign(V) Y_0(\lambda \pm i0)
\]
for any $\lambda\in \R$.

We shall single out energies $\lambda\in\R$, where the following Lipschitz estimate holds true:
\begin{equation}\label{LipAssump}
\exists L>0 \, \forall 0 <  \mu\leq 1:\quad \bigl\|Y_0(\l\pm  i \mu)- Y_0(\l\pm i 0)\bigr\|\le L \mu.
\end{equation}
We write
\[
\Lips(H) = \bigl\{\lambda\in\R \,| \, \textup{\eqref{LipAssump} is satisfied at } \lambda \bigr\}.
\]
Clearly $\R\setminus \sigma(\Delta) = \R\setminus [-d,d]\subseteq  \Lips(H)$, and
we may write $\Lips(H)$ as in increasing union of subsets
\[
\Lips(H;L) = \bigl\{\lambda\in\Lips(H) \,| \, \textup{\eqref{LipAssump} is satisfied at } \lambda \textup{ with constant } L \bigr\}.
\]
The sets $\Lips(H;L)$ are closed, since $\lambda\to Y_0(\lambda\pm i \mu)$ are continuous for $\mu\geq 0$. It is obscured by the choice of notation, that the sets $\Lips(H)$ and $\Lips(H;L)$ may depend on the choice of $q_2$ made in the factorization of $V$. 

\begin{lemma}\label{lem:BS} Let $L>0$ and suppose $\lambda \in\sigma_\BS(H)\cap\Lips(H;L)$ and $f$ a solution of \eqref{BirSch}.
Then $q_2 f \in D((\D-\l)^{-1})$ and $\|(\D-\l)^{-1}q_2 f\|_2\leq L \|f\|_2$. Furthermore,
\[
(\D-\l)^{-1} q_2 f =
 \lim_{\mu\to 0 }R_0(\l+i \mu)^{-1} q_2 f.
 \]
\end{lemma}

\begin{proof} Let $f\in\ell^2(\Z^d)$ be a solution of \eqref{BirSch} at energy $\lambda\in\R$.
Denote by $E_{q_2 f}$ the spectral  measure for $\D$ associated with the state $q_2 f$.
Then $q_2 f\in D((\Delta-\l)^{-1})$ if and only if
$\int_{\R}(x-\l)^{-2} d E_{q_2 f}(x) <\infty$.
Compute for $\mu > 0$
\[\label{LipEstHelp}
\begin{aligned}
& \int_{\R}((x-\l)^2 + \mu^2)^{-1} d E_{q_2 f}(x) =
 (q_2 f, R_0(\l- i\mu)R_0(\l + i \mu) q_2 f)\\
& \quad  =  \frac1{2\mu}\bigl(f, (\wt Y_0(\l- i\mu)-\wt Y_0(\l + i \mu) ) f\bigr)\\
& \quad  =   \frac1{2\mu}\bigl(f, (\wt Y_0(\l- i\mu)-\wt Y_0(\l - i 0) ) f\bigr)
  + \frac1{2\mu} \bigl(f, (\wt Y_0(\l+ i 0)-\wt Y_0(\l + i \mu) ) f\bigr) \\
& \qquad  +  \frac1{2\mu}\Im \bigl\{\bigl(f, \wt Y_0(\l + i 0)  f\bigr)\bigr\}.
\end{aligned}
\]
Note that by the Birman-Schwinger equation \eqref{BirSch}, as well as Eqs. \eqref{SuppOfBSsol} and \eqref{tYtoY}:
\[
\begin{aligned}
\Im\bigl\{\bigl(f, \wt Y_0(\lambda + i 0) f\bigr)\bigr\} &= 
\Im \bigl\{\bigl(\sign(V)f,  Y_0(\lambda + i 0)  f\bigr)\bigr\}\\
 &   = - \Im \bigl\{\bigl(\sign(V) f, f\bigr)\bigr\}= 0.
\end{aligned}
\]
The result now follows from \eqref{LipAssump}, \eqref{LipEstHelp}, and the monotone convergence theorem.

As for the claimed identity, we need to argue that
$\lim_{\mu\to 0}  R_0(\lambda+i \mu)^{-1} q_2 f = (\Delta-\lambda)^{-1} q_2 f$.
But this follows from the computation
\[
R_0(\lambda+i \mu)q_2 f - (\Delta-\lambda)^{-1} q_2 f=
 i \mu R_0(\lambda+i \mu)(\Delta-\lambda)^{-1}q_2 f,
\]
together with Lemma~\ref{lem:TrivialLimit}.
\end{proof}

\begin{lemma}\label{Lemma-BirSch} We have $\sigma_\BS(H)\cap\Lips(H)\subseteq \sigma_\pp(H)$. Furthermore, for $\lambda \in \sigma_\BS(H)\cap\Lips(H)$,
the linear map $f\to (\D-\lambda)^{-1} q_2 f$ is well-defined and takes nonzero solutions of \eqref{BirSch} at $\lambda$
into nonzero eigenfunctions of $H$ with eigenvalue $\lambda$.
\end{lemma}

\begin{proof}
 Suppose $f\neq 0$ solves the Birman-Schwinger equation \eqref{BirSch}. Then, by Lemma~\ref{lem:BS},
 $q_2 f\in D((\D-\lambda)^{-1})$ and we may put $g = (\D- \lambda)^{-1} q_2 f$.
 By the spectral theorem, we have
 \[
 (H-\lambda) g = q_2 f + V g = - q_2  Y_0(\lambda + i 0) f + V g  = 0,
 \]
 where we used the identity
 \[
 Y_0(\lambda+i 0) f
 =
 \lim_{\mu\to 0_+} q_1 R_0(\lambda+i \mu)^{-1} q_2 f =
  q_1(\D-\lambda)^{-1} q_2 f = q_1 g
 \]
  in the last step. It remains to argue that $g\neq 0$. But $f = -q_1 g\neq 0$ by the above
  identity and hence, $g\neq 0$.
\end{proof}

\begin{proposition}\label{prop:BS-compact} $\sigma_\BS(H)$ is a compact subset of $\sigma(H)$ with zero Lebesgue measure.
\end{proposition}

\begin{proof} Since $\R\setminus \sigma(H) \subseteq \Lips(H)$, we conclude from Lemma~\ref{Lemma-BirSch}
that $\sigma_\BS(H)\subseteq \sigma(H)$. That $\sigma_\BS(H)$ is a closed set (and hence compact) follows
from the observation that  $\lambda \to Y_0(\lambda+ i 0)$ is continuous with values in compact operators.

To see that the measure of $\sigma_\BS(H)$ is zero, we follow an argument from the proof of
\cite[Lemma~4.20]{KaKu71}. Let $\lambda\in \sigma_\BS(H)$. Since $Y_0(\lambda+i0)$ is compact, there exists a circle
$\Gamma_\lambda$ enclosing $-1$ in the complex plane such that $\Gamma_\lambda\subseteq \rho(Y_0(\lambda+i0))$ -- the resolvent set of $Y_0(z)$ -- 
and $0$ is in the unbounded connected component of $\C\sm \Gamma_\lambda$.
By continuity of $\overline{\C}_+ \ni z\to Y_0(z)$, we deduce the existence of $r_\lambda>0$, such that
$\Gamma_\lambda\subseteq \rho(Y_0(z))$ for $z\in \overline{D_\lambda}$, where $D_\lambda := \{z\in C_+ \, |\,  |z-\lambda| < r_\lambda\}$. Here and below we sometimes abuse notation and write, e.g.,  
$Y_0(z)$ for $Y_0(z+i0)$ when $z\in\R$.

Define for $z\in \overline{D_\lambda}$ the finite rank Riesz projection
\[
P_z = \frac1{2\pi i} \int_{\Gamma} (w-Y_0(z))^{-1} \, dw
\]
and observe that $\rank(P_z) = \rank(P_{\lambda+i0}) =: n_0$ is constant throughout $\overline{D_\lambda}$.
By possibly choosing $r_\lambda$ smaller, we may assume that
$\|P_z-P_{\lambda+i0}\|\leq 1/2$ for $z\in \overline{D_\lambda}$.

 Abbreviate $\cH_z = P_z\ell^2(\Z^d)$ for $z\in \overline{D_\lambda}$. Let $\Pi\colon \C^{n_0}\to \cH_{\lambda+i0}$ be a linear isomorphism. 
Define $\Theta_z = {P_z}_{| \cH_{\lambda+i0}}\colon \cH_{\lambda+i0} \to \cH_z$, which is a linear isomorphism with left inverse
\[
\Theta_z^{-1} = \bigl(1 + P_{\lambda+i0}(P_z-P_{\lambda+i0})\bigr)^{-1}P_{\lambda+i0}\colon \cH_z\to \cH_{\lambda+i0}.
\]

We define a family of linear operators on $\C^{n_0}$ by setting
\[
X_0(z) = \Pi^{-1} \Theta_z^{-1}(I+Y_0(z))\Theta_z \Pi
\]
for $z\in \overline{D_{\lambda}}$. Note that $z\to X_0(z)$ is
holomorphic in $D_\lambda$ and continuous in the closure.
%Note that $P_z$ is nonzero since $-1\in\sigma(Y(\lambda+i0))$ and %$P_z$ depends continuously on $z\in \overline{D_\lambda}$. 
Furthermore, for $z\in  \overline{D_\lambda}$, we have $-1\in \sigma(Y_0(z))$ if and only if
$0\in \sigma(X_0(z))$.

Denote by $\psi\colon \dD \to D_\lambda$ a conformal equivalence between the unit disc $\dD$ and $D_\lambda$.
Note that $\psi$ extends by continuity to a homeomorphism $\psi \colon \overline{\dD}\to \overline{D_\lambda}$. (See \cite[Cor.~17.18]{RSFM15} applied to $\psi^{-1}$.) 
Then $\varphi(\xi) := \det(X_0(\psi(\xi)))$ defines a continuous function on $\overline{\dD}$, holomorphic in $\dD$.

We now invoke \cite[Thm.~13.20]{RSFM15} to conclude that $\varphi$ at most vanishes on a subset $M$ of 
$\T = \partial\dD$ of zero Lebesgue measure, provided $\varphi$ is not identically zero in $\dD$. 
To see that this is the case pick $\xi\in\dD$ and assume towards a contradiction that
$\varphi(\xi)=0$. Let $z = \psi(\xi)\in D_\lambda$, and observe that $-1$ must be an eigenvalue of
$Y_0(z)$. Denote by $f$ an eigenfunction and compute
\[
\begin{aligned}
0 & = -\Im\bigl\{\bigl(\sign(V)f,f\bigr)\} = \Im\bigl\{ \bigl(\sign(V)f,Y_0(z) f\bigr)\bigr\}\\
&  = \Im\bigl\{\bigl(q_1 \sign(V)f,(\Delta-z)^{-1}q_2f\bigr)\bigr\}
 =  \Im\bigl\{\bigl(q_2f,(\Delta-z)^{-1}q_2f\bigr)\bigr\}\\
 & = \Im(z)\bigl(q_2f, ((\Delta-\Re z)^2+ \Im(z)^2 )^{-1} q_2 f \bigr)
\geq \frac{\|q_2 f\|^2_2}{\Im(z)},
\end{aligned}
\]
which is a contradiction unless $f=0$. 

Since $\psi(M) = \sigma_\BS(H)\cap [\lambda-r_\lambda,\lambda+r_\lambda]$, and $\psi$'s extension to the boundary maps sets of measure zero into sets of measure zero, we conclude the proof by a compactness argument. (Note that $w_\pm  = \psi^{-1}(\lambda\pm r_\lambda)$ splits $\T$ into two open arcs and $\psi$'s extension across these arcs are in fact holomorphic by Schwarz' reflection principle.)
\end{proof}

\begin{lemma}\label{lemma:BS-Finite} For any $L>0$, the set $\sigma_{\BS}(H)\cap \Lips(H;L)$ is finite.
\end{lemma}

\begin{proof} Suppose towards a contradiction that there  exists a countable sequence of distinct
real numbers $\{\lambda_n\}\subseteq\sigma_\BS(H)\cap \Lips(H;L)$. Let  $\{f_n\}_{n=1}^\infty$ be an associated sequence of normalized solutions of \eqref{BirSch}. By Proposition~\ref{prop:BS-compact} and closedness of $\Lips(H;L)$, we may assume that
$\lambda_n\to \lambda \in \sigma_\BS(H)\cap \Lips(H;L)$.

By Banach-Alaoglu's theorem, we may extract a subsequence $\{f_{n_k}\}_{k\in\N}$, such that
$f_{n_k}\to f$ weakly. Since $Y_0(\lambda) = q_1 R_0(\lambda + i 0) q_2$ is compact,
$Y_0(\lambda)f_{n_k}\to Y_0(\lambda) f$ in norm. Since $\lambda\to Y_0(\lambda)$ is continuous we may finally conclude that
$f_{n_k} = - Y_0(\lambda_{n_k})f_{n_k} \to -Y_0(\lambda) f$ in norm. Hence $f$ is a normalized solution of \eqref{BirSch}
at $\lambda$.

We may now use \eqref{Lemma-BirSch} to construct a sequence of eigenfunction $g_k = (\D-\lambda_{n_k})^{-1} q_2 f_{n_k}$ for $H$, all satisfying $\|g_k\|_2\leq L$ due to Lemma~\ref{lem:BS}. Compute

\[
(q_1 f,g_k) =  (f, q_1R_0(\lambda_{n_k}+i 0) q_2 f_{n_k} ) = - (f,f_{n_k}).
\]
We have arrived at a contradiction, since $g_k \to 0$ weakly -- being an orthogonal uniformly bounded sequence -- and $f_{n_k}\to f$ in norm.
\end{proof}

Recall the notation $\cB_2$ for the class of Hilbert-Schmidt operators on $\ell^2(\Z^d)$, and the exponent $\gamma_{d,q}$ from \eqref{gamma-def}.

\begin{theorem}\label{LimAbs}
Put $Y(z) := q_2(H-z)^{-1}q_2$ for $z\in\C\sm \sigma(H)$. Suppose $d\geq 3$ and
$V\in \ell^p(\Z^d)$ with $p$ satisfying \eqref{Vc}. Then:
\begin{enumerate}[label= \textup{(\alph*)}]
\item\label{Item-LimAbs-a} For any compact set $J\subseteq \R\sm\sigma_\BS(H)$, we have
\begin{equation}\label{Eq-LAP}
\sup_{\lambda\in J, \mu\neq 0} \|Y(\lambda+i\mu)\|_{\cB_2} <\infty.
\end{equation}
\item\label{Item-LimAbs-a2} Let $\cO\subseteq \C$ be an open set with $\sigma_\BS(H)\subseteq \cO$, and
let $\gamma > \gamma_{d,2p}$. Then there exists $C>0$, such that for all $z,z'\in \C \sm ([-d,d]\cup \cO)$ with $\Im z\Im z' \geq 0$, we have
\[
\|Y(z) - Y(z')\|_{\cB_2} \leq C|z-z'|^\gamma.
\]   
\item\label{Item-LimAbs-b} We have $\sigma_\mathrm{sc}(H)\subseteq \sigma_\BS(H)$.
\item\label{Item-LimAbs-c} Let $P_\mathrm{ac}$ denote the orthogonal projection onto the absolutely continuous subspace $\cH_\mathrm{ac}$ pertaining to $H$. Then the wave operators
\[
W^\pm := \slim_{t\to \pm\infty} e^{itH}e^{-it\Delta}
\qquad \textup{and}\qquad \widetilde{W}^\pm :=  \slim_{t\to \pm\infty} e^{it\Delta}e^{-it H}P_\mathrm{ac}
\]
exist, $(W^\pm)^* = \widetilde{W}^\pm$, $W^\pm \Delta = H W^\pm$
and $W^\pm\ell^2(\Z^d) = \cH_\mathrm{ac}$.
\end{enumerate} 
\end{theorem}

\begin{proof} For the purpose of this proof we choose $q_2$, such that $q_2(n)\neq 0$ for all $n\in \Z^d$.
Abbreviate $\Omega= \R\sm\sigma_\BS(H)$ and note that $\Omega$ is an open set. Recall the notation $Y_0(z)$ and $\wt Y_0(z)$ from 
\eqref{YandwtY}.

To establish \ref{Item-LimAbs-a}, we first compute for $z$ with $\Im z\neq 0$:
\[
Y(z) Y_0(z) = q_2(H-z)^{-1}q_2q_1(H_0-z)^{-1}q_2 =
q_2(H_0-z)^{-1}q_2 - q_2(H-z)^{-1}q_2.
\]
Hence
\[
Y(z)\bigl(I+Y_0(z)\bigr) = q_2(H_0-z)^{-1}q_2.
\]
Let $\lambda\in J$. Then, by continuity, there exists $r_\lambda>0$ and $C_\lambda$ such that $I+Y_0(z)$ is bounded invertible for $|z-\lambda|\leq r_\lambda$ with $\Im z\neq 0$
and $\|(I+Y_0(z))^{-1}\|\leq C_\lambda$. Here we used Remark~\ref{BSminus}, which ensures invertibility also for $z$ with $\Im z<0$. By Theorem~\ref{T2HS} and compactness of $J$, there exists $r>0$, such that the claimed bound holds for $0<|\mu| \leq r$, which clearly suffices. Here we used that if 
$S$ is Hilbert-Schmidt and $T$ is bounded, then $ST$ is Hilbert-Schmidt and $\|ST\|_{\cB_2} \leq \|S\|_{\cB_2}\|T\|$.

The claim \ref{Item-LimAbs-a2} follows in a similar fashion from Theorem~\ref{T2HS} and the computation
\[
\begin{aligned}
Y(z) - Y(z') & = 
\wt Y_0(z)(I+Y_0(z))^{-1}\bigl(Y_0(z')-Y_0(z)\bigr)(I+Y_0(z'))^{-1}  \\
& \qquad + \bigl(\wt Y_0(z)-\wt Y_0(z')\bigr)\bigl(I+Y_0(z')\bigr)^{-1}.
\end{aligned}
\]

We now turn to \ref{Item-LimAbs-b}. It follows from \eqref{Eq-LAP} and \cite[Thm.~XIII.20]{RS78}
that for any bounded open interval $(a,b)$ with $J = [a,b]\subseteq  \Omega$, we have
$E_{(a,b)}(H)\ell^2(\Z^d)\subseteq \cH_\mathrm{ac}$. Here we used that $q_2$ was chosen to be nowhere vanishing.

Let $\Omega_n$ be a sequence of finite unions of disjoint intervals of the form $(a,b)$ considered above,
and chosen such that $\Omega_n\subseteq\Omega_{n+1}$ and $\cup_{n=1}^\infty \Omega_n = \Omega$.
Let $f\in \ell^2(\Z^d)$. It follows that $E_{\Omega}(H)f = \lim_{n\to\infty} E_{\Omega_n}f\in \cH_\mathrm{ac}$, and completes the proof of \ref{Item-LimAbs-b}.

Finally we verify \ref{Item-LimAbs-c}.
First note that due to \ref{Item-LimAbs-b} and Proposition~\ref{prop:BS-compact}, we may conclude that 
$E_{\sigma_\BS(H)}(H)\ell^2(\Z^d) = \cH_\mathrm{sc}\oplus\cH_\mathrm{pp}$ (recall \eqref{SpecDecomp}).
Consequently, $E_\Omega(H) = P_\mathrm{ac}$, the orthogonal projection onto the absolutely continuous subspace.

Next we observe that by \cite[Thm.~XIII.31]{RS78}, the reduced wave operators
\[
W_n^\pm := \slim_{t\to \pm\infty} e^{itH}e^{-it\Delta} E_{\Omega_n}(\D)
\qquad \textup{and}\qquad \widetilde{W}_n^\pm :=  \slim_{t\to \pm\infty} e^{it\Delta}e^{-it H}E_{\Omega_n}(H)
\]
exist for each $n$, $(W_n^\pm)^* = \widetilde{W}_n^\pm$, $W_n^\pm \Delta = H W_n^\pm$
and $W_n^\pm\ell^2(\Z^d) = E_{\Omega_n}(H)\ell^2(\Z^d)$.

We may now conclude that the wave operators exist, and we have the relations $\slim_{n\to \infty} W_{n}^\pm = W^\pm$ and  $\slim_{n\to \infty} \widetilde{W}_{n}^\pm = \widetilde{W}^\pm$, where we used that $E_\Omega(H) = P_\mathrm{ac}$.
 From this the remaining claims follow.
\end{proof}

We end this section with:

\begin{proof}[Proof of Theorem~\ref{T3}~\ref{Item-T3-a1},~\ref{Item-T3-a2} and \ref{Item-T3-b}]
We begin with \ref{Item-T3-a1}. That the closure of the set $\sigma_\mathrm{sc}(H)\cup\sigma_\mathrm{pp}(H)$ has zero Lebesgue measure follows from Lemma~\ref{lem:From-pp-to-BS}, Theorem~\ref{LimAbs}~\ref{Item-LimAbs-b}, and Proposition~\ref{prop:BS-compact}.
That $\sigma_\mathrm{pp}(H)=\sigma_\BS(H)=\emptyset$ if
$\|V\|_p < (1+C_d^0 \Gamma(r,d,0))^{-1}$ is a consequence of the observation that $\sigma_\BS(H)=\emptyset$, which follows directly from \eqref{VRV1}. 
Here we used that
$\sigma_\BS(H)$ does not depend on $q_2$, so that we may choose $q_2 = \sign(V) q_1$ for which
$\|q_1\|_q\|q_2\|_q = \|V\|_p$. 

The asymptotic completeness statement in \ref{Item-T3-a2} is a part of Theorem~\ref{LimAbs}~\ref{Item-LimAbs-c}. 

 The claim in \ref{Item-T3-b} that $\sigma_\mathrm{pp}(H)$ is finite under the assumed conditions, follows by combining Lemma~\ref{lem:From-pp-to-BS} with Lemma~\ref{lemma:BS-Finite}, keeping in mind Remark~\ref{RemLip}~\ref{Item-RemLip-b}. The lemma implies that $\sigma_\BS(H)$ is a finite set. The absence of singular continuous spectrum now follows from
Theorem~\ref{LimAbs}~\ref{Item-LimAbs-b}.
\end{proof}

\begin{appendix}

\section{Pointwise estimates of Bessel functions}\label{AppA}
\setcounter{equation}{0}

Recall the following properties of the Bessel function \eqref{BesselFunction} valid for all $(t,n)\in \R\ts \Z$:
\begin{equation}\lb{Be}
{2n\/t}J_n(t)=J_{n+1}(t)+J_{n-1}(t),
\qq J_{-n}(t)= (-1)^nJ_n(t) \qq \textup{and} \qq  J_{n}(-t)=(-1)^nJ_n(t).
\end{equation}

The following estimate on $J_0$ is due to Szeg\H{o} \cite{Sz75}:
\[\label{Landau0}
|J_0(t)|\leq \sqrt{\frac2{\pi|t|}}.
\]

We shall make use of two optimal universal estimates on Bessel functions due to Landau.

\begin{lemma}[\cite{La00}]\label{Landau} We have the following pointwise bounds for all real $n,t$:
\[\label{LandauB}
|J_n(t)|\leq  B_\Lan|n|^{-\frac13},
\]
and
\[\label{LandauC}
|J_n(t)|\leq  C_\Lan|t|^{-\frac13},
\]
where $B_\Lan < 7/10$ and $C_\Lan <  4/5$.
\end{lemma}

One can find Landau's optimal constants $B_\Lan$ and $C_\Lan$ with several decimals in \cite{La00}.
Secondly, we exploit another universal estimate due to Krasikov.

\begin{lemma}[\cite{Kr14}]\label{lemma-Krasikov} We have the following pointwise bound for $(n,t)\in\R\times \R$ with
$n\geq 1/2$ and $t\geq 0$:
\[\label{Krasikov}
|J_n(t)|\leq  \sqrt{\frac2{\pi}} \frac1{|t^2-|n^2-\frac14||^\frac14}.
\]
\end{lemma}

\begin{proposition}\label{FusedBessel} For any integer $n\in\Z$ and $t\in\R$ with
$|t|\geq 1$:
\[\label{EqFused}
|J_n(t)|\leq \frac{1}{|t|^\frac14(|n|^\frac13 + ||t|-|n||)^\frac14},
\]
(The estimate remains valid for non-integer $n$ with $n\geq 1$.)
\end{proposition}

\begin{proof} By \eqref{Be}, it suffices to show the estimate for
$t\geq 1$, $n\geq 1$. Note that for $n=0$, the estimate \eqref{Landau0} implies \eqref{EqFused}.

We estimate first supposing $|t-n|\leq n^{1/3}$:
\[\label{TransReg}
\frac{t^\frac14(n^{\frac13}+|t-n|)^\frac14}{n^{\frac13}} \leq
\frac{(n+n^\frac13)^\frac14(2n^{\frac13})^\frac14}{n^{\frac13}}
=2^\frac14 \bigl(1+n^{-\frac23}\bigr)^\frac14 \leq \sqrt{2}.
\]

As for the regime $|t-n| \geq n^{1/3}$ observe first that
\[
\frac{1}{\bigl|t^2-(n^2-\frac14)\bigr|^\frac14} \leq \frac1{t^\frac14\bigl|t-\sqrt{n^2-\frac14}\bigr|^\frac14}.
\]
Secondly, let $n\geq 1$ and $t\geq n+n^{1/3}$. Observe that
\[
t\to \frac{n^\frac13 + t -n }{t- \sqrt{n^2-\frac14}}
\]
is decreasing (towards $1$) and hence
\[
\frac{(n^\frac13 + t -n )}{t- \sqrt{n^2-\frac14}}
\leq \frac{2n^{\frac13}}{n+n^\frac13 - \sqrt{n^2-\frac14}}
\leq 2.
\]
Thirdly, for $1\leq t \leq n-n^{1/3}$ (hence $n\geq 3$), we have
\[
t \to \frac{n^\frac13 + n- t }{\sqrt{n^2-\frac14}-t}
\]
is increasing and therefore
\[
\frac{n^\frac13 + n- t }{\sqrt{n^2-\frac14}-t}
\leq \frac{2n^\frac13  }{\sqrt{n^2-\frac14}-n+n^\frac13}.
\]
Observe that the right-hand side converges to $1$ as $n\to\infty$.
We claim that the right-hand side does not exceed $2$ for $n\geq 1$. If it does, then by continuity the equation
\[
\frac{2n^\frac13  }{\sqrt{n^2-\frac14}-n+n^\frac13} = 2
\]
must have a solution $n_0\geq 1$ (that may not be integer). Then we would have $\sqrt{n_0^2-\frac14} = 2 n_0$, which is absurd.

To sum up, for $n\geq 1$ and $t\geq 1$ with $|t-n|\geq n^{1/3}$
\[
\frac1{|t^2-(n^2-\frac14)|^\frac14} \leq \frac{2^\frac14}{t^\frac14(n^\frac13 + |t-n|)^\frac14}.
\]
Recalling \eqref{TransReg},
the desired estimate \eqref{EqFused} now follows from \eqref{LandauB} and \eqref{Krasikov}, since $\sqrt{2} B_\mathrm{L}\leq 1$ and
$2^{1/4} \sqrt{2/\pi}\leq 1$.
\end{proof}

\begin{lemma}\label{lem-small-t} For $n\in\Z$ and $t\in [-1,1]$, we have
\[
|J_n(t)| \leq \frac{1}{(|n|+1)^\frac12}.
\]
\end{lemma}

\begin{proof}
We may again assume that $0\leq t\leq 1$ and $n\geq 1$. For $n=0$ the estimate is trivial, since $|J_n(t)|\leq 1$ for all $t\in \R$ and integer $n$.

For $n=1,2$, we use \eqref{LandauB} to estimate
$|J_n(t)| \leq B_\Lan n^{-1/3} \leq B_\Lan 2^\frac16 n^{-1/2}$.

If $n\geq 2$, we estimate using \eqref{lemma-Krasikov}:
$|t-(n^2-1/4)| \geq n^2-5/4 \geq 2 n^2/3$
and hence
\[
|J_n(t)| \leq \sqrt{\frac2{\pi}} \frac1{|t- |n^2-\frac14||^\frac14}
\leq \frac{6^\frac14}{(\pi n)^\frac12}.
\]
This completes the proof since the prefactor in both cases is smaller than $1$.
\end{proof}

\section{Weighted $L^p$-estimates on Bessel functions}\label{AppB}

We have the following weighted $\ell^p(\Z^d)$ estimate on Bessel functions. This improves on an estimate of Stempak \cite[Eq. (3)]{St}, and in particular matches the asymptotic presented in \cite[Eq. (6)]{St}.

\begin{lemma}\lb{TA1}
For all  $\g \in [0,1]$, $p>2+2\g$ and $n\in \Z$ the following
estimate hold true:
\begin{equation}\label{es1}
\int_1^\iy t^\g |J_n(t)|^p dt \le  C_p^\gamma \vk_p^\gamma(n)^p,
\end{equation}
where
\[\label{Cpgamma}
C_p^\gamma=\ca  8\bigl({1\/p-2-2\gamma}+{1\/4-p}  \bigr),  & 2+2\gamma <p<4\\
              \frac4{1-\gamma}, & p=4\\
 \frac{14 \cdot 2^\frac{p}4}{p-4},  & p>4,
 \ac
 \]
 $\vk^\gamma_p(0)=1$ and  $\vk^\gamma_p(n)$ for $n\neq 0$ is given by
\[\label{kpgamma}
\vk^\gamma_p(n)=\begin{cases} |n|^{-{1\/2}+{1+\g\/p}}, & 1<p<4\\
                        |n|^{-{1-\g\/4}}(1+\log |n|)^{1\/4}, & p=4\\
                       |n|^{-{1\/3}+{1\/3p}+{\g\/p}}, & p>4.
\end{cases}
\]
\end{lemma}

\begin{proof} Using Proposition~\ref{FusedBessel}, we have
\[
\int_1^\infty t^\gamma |J_n(y)|^p \, dt
\leq \int_1^\infty \frac{t^{\gamma-\frac{p}4}}{(n^\frac13 + |t-n|)^\frac{p}4}\, dt
= n^{1+\gamma-\frac{p}2} \int_\delta^\infty
\frac{s^{\gamma-\frac{p}4}}{(\delta+|s-1|)^\frac{p}{4}} \, ds,
\]
where $\delta = n^{-2/3}$. We split the intergal on the right-hand side into a sum of two integrals
\[
I_p^\gamma = \int_\delta^1\frac{s^{\gamma-\frac{p}4}}{(\delta+|s-1|)^\frac{p}{4}} \, ds \quad \textup{and} \quad
J_p^\gamma= \int_1^\infty \frac{s^{\gamma-\frac{p}4}}{(\delta+|s-1|)^\frac{p}{4}} \, ds.
\]
We now proceed to estimate these two integrals in three different case.

\emph{Case I:} $2+2\gamma < p < 4$, where $\gamma - p/4 <0$.
Estimate first
\[
\begin{aligned}
J_p^\gamma & \leq \int_1^2 \frac1{(s-1)^{\frac{p}4}}\, ds + \int_2^\infty \frac1{(s-1)^{\frac{p}2-\gamma}}\,ds\\
& = \frac{2^{1-\frac{p}4}-1}{1-\frac{p}4}+\frac1{1+\gamma- \frac{p}2} \leq 2 + \frac{2}{p-2-2\gamma}\leq \frac{6}{p-2-2\gamma}.
\end{aligned}
\]
As for $I_p^\gamma$ we estimate
\[
I_p^\gamma\leq \int_0^1 \frac1{s^{\frac{p}4}(1-s)^\frac{p}4}\,ds
\leq  2^{1+\frac{p}4}\int_0^{\frac12} s^{-\frac{p}4}\, ds = \frac{2^{\frac{p}2}}{1-\frac{p}4}\leq \frac{8}{4-p}.
\]

\emph{Case II:} $p=4$ and $0\leq \gamma < 1$.

Here we estimate
\[
J_p^\gamma \leq \int_1^2 \frac{1}{\delta+ s-1}\, ds + \int_2^\infty \frac1{(s-1)^{2-\gamma}}\, ds
\leq \ln(2) -\ln(\delta)+ \frac1{1-\gamma}
\]
and
\[
I_p^\gamma \leq  \int_\delta^1 \frac1{s(\delta+1-s)}\, ds
= \int_{\frac{\delta}2}^{1-\frac{\delta}2} \frac1{(r+\frac{\delta}{2})(\frac{\delta}2 + 1 -r)} \, dr \leq
4 \int_{\frac{\delta}{2}}^{\frac12} \frac1{r+\frac{\delta}2} \, dr
\leq -4 \ln(\delta).
\]
This implies the claimed estimate, since
$\ln(2)\leq (1-\gamma)^{-1}$ and $-5\ln(\delta) = \frac{10}3\ln(n)$.

\emph{Case III:} $p>4$. We again estimate, using that $p/4-\gamma>0$,
\[
J_p^\gamma \leq \int_1^2 \frac{1}{(\delta+1-s)^{\frac{p}4}} \, ds
+ \int_2^\infty \frac{1}{(s-1)^{\frac{p}2-\gamma}}\, ds
\leq \frac{4}{p-4}\delta^{1-\frac{p}4}
+\frac2{p-2-2\gamma}
\]
and similarly to previous estimates of $I_p^\gamma$
\[
I_p^\gamma \leq \int_\delta^1 \frac{1}{s^{\frac{p}4}(\delta+1-s)^{\frac{p}4}}\, ds
%=2 \int_{\frac{\delta}2}^\frac12 %\frac1{(s+\frac{\delta}{2})^{\frac{p}4}(\frac{\delta}{2}+1-s)^{\frac{p}4}}\,
\leq 2^{1+\frac{p}4}\int_{\frac{\delta}2}^{\frac12}
\frac1{(s+\frac{\delta}{2})^\frac{p}4}\, ds
\leq \frac{8 \cdot 2^\frac{p}4}{p-4} \delta^{1-\frac{p}4}.
\]
This implies the remaining estimate.

Finally we must check the bounds with $n=0$. Here we use \eqref{Landau0} and find that
\[
\int_1^\infty t^\gamma |J_0(t)|^p \, ds\leq \sqrt{\frac{2^p}{\pi^p}}
\int_1^\infty t^{-\frac{p}2+\gamma}\, ds
\leq \frac2{p-2-2\gamma}\leq C_p^\gamma.
\]
This completes the proof.
\end{proof}

The following lemma will be used in conjunction with the weighted $L^p$-estimate from Lemma~\ref{TA1}.

\begin{lemma}\label{lemma-kappa-r-bound} Let $\gamma\in[0,1]$
and $p>2+2\gamma$. Define
\[\label{r-p-gamma}
r_p^\gamma = \ca \frac{2p}{p-2-2\gamma}, & 2+2\gamma < p \leq 4\\
\frac{3p}{p-1-3\gamma}, & p>4\ac.
\]
If $r_p^\gamma<r\leq \infty$, then the sequence $\{\kappa_p^\gamma(n)\}_{n\in\Z}$ is an element of $\ell^r(\Z)$ and
if $r\neq \infty$ we have
\[
\|\kappa_p^\gamma\|_r \leq \ca \Bigl(3+ \frac2{\frac{r}{r_p^\gamma}-1}\Bigr)^\frac1{r}, & p\neq 4,\\
 \Bigl( 3+ 2\bigl( \frac{1+\frac{r}4}{\frac{r}{r_4^\gamma}-1}
\bigr)^{1+\frac{r}{4}}\Bigr)^\frac1{r} , & p=4.\ac.
\]
We furthermore have $\|\kappa_p^\gamma\|_\infty = 1$ if $p\neq 4$
and $\|\kappa_4^\gamma\|_\infty \leq (1-\gamma)^{-1}$.
\end{lemma}

\begin{proof} Note first that for any $p>2+2\gamma$ and $r>0$, we have
\[
\|\kappa_p^\gamma\|_r^r  = 3 + 2\sum_{n=2}^\infty \kappa^\gamma_p(n)^r.
\]

For $p\neq 4$, the lemma follows easily from the estimate and computation
\[
\sum_{n=2}^\infty \kappa^\gamma_p(n)^r\leq \int_1^\infty \kappa_p^\gamma(s)^r \, d s
= \frac1{\frac{r}{r_p^\gamma}-1}.
\]

For $p=4$,
we use the estimate $1+\log(x)\leq \frac{e^{\sigma-1}}{\sigma}x^\sigma\leq x^\sigma/\sigma$ valid for $x\geq 1$ and $0<\sigma\leq 1$, and obtain
\[
\sum_{n=2}^\infty \kappa^\gamma_4(n)^r
\leq \sigma^{-\frac{r}4} \int_1^\infty s^{-\frac{r}{r^\gamma_4}+\frac{\sigma r}{4}}\, ds = \sigma^{-\frac{r}4}\cdot \frac{1}{\frac{r}{r^\gamma_4}-\frac{\sigma r}{4} -1},
\]
provided $\sigma < 4(r/r^\gamma_4-1)/r\leq 1$. The right-hand side is minimized by choosing
\[
\sigma = \frac{\frac{r}{r^\gamma_4}-1}{1+\frac{r}4}.
\]
This choice gives
\[
\sum_{n=2}^\infty \kappa^\gamma_4(n)^r
\leq \Bigl( \frac{1+\frac{r}4}{\frac{r}{r_4^\gamma}-1}
\Bigr)^{1+\frac{r}{4}}
\]
and completes the proof
%Integration by substitution, yields for $\rho>0$ and $\sigma..$:
%\[
%\int_1^\infty s^{-(\rho+1)}(1+\log(s))^\sigma \, ds = e^{\rho}\int_1^\infty e^{-\rho y} y^\sigma \, dy
%\leq \frac{e^\rho}{\rho^{\sigma+1}} \Gamma(\sigma+1) =   \frac{\sigma e^\rho}{\rho^{\sigma+1}}\Gamma(\sigma).
%\]
%This yields
%\[
% \int_1^\infty s^{-\frac{r(1-\gamma)}{4}}(1+\log(s))^{\frac{r}{4}}\, ds
% \leq \frac{r e^{\frac{r(1-\gamma)}{4}-1}}{4(\frac{r(1-\gamma)}{4}-1)^{(r+4)/4}}\Gamma(r/4) =  \frac{r %e^{\frac{r}{r^\gamma_4}-1}}{4(\frac{r}{r^\gamma_4}-1)^{(r+4)/4}}\Gamma(r/4)
%\]
\end{proof}

\section{Various discrete estimates}\label{AppC}

For the readers convenience we recall first two well-known estimates, the first of which can be found in \cite[Theorem~IX.17]{RS75} and the proof of the second is just a repetition of the proof in the more usual continuous case 
\cite[Theorem~4.2]{LL97}.
%For a positive sequence $\om=(\om_n)_{n\in \Z^d}$, with %$\om_n>0$, we define weighted spaces
%\[
%\lb{defsp}
%\ell_\om^{p}(\Z^d)=\rt\{f=(f_n)_{n\in \Z^d}: \|f\|_{\ell_{\om}
%^{p}(\Z^d)}^p=\sum_{n\in \Z^d}\om_n |f_n|^p<\iy    \rt\},\qq p\ge %1.
%\]

\begin{theorem}[Discrete Riesz-Thorin]
\lb{TRT}
Let  $1\le p_0,p_1,q_0,q_1\le \iy$
and suppose that $T$ is a linear operator from $\ell^{p_0}(\Z^d)\cap \ell^{p_1}(\Z^d)$ to  $\ell^{q_0}(\Z^d)\cap \ell^{q_1}(\Z^d)$,
which satisfies
\[
\|Tf\|_{q_0}\le M_0\|f\|_{p_0}\quad \textup{and} \quad \|Tf\|_{q_1}\le M_1\|f\|_{p_1}.
\]
Then for each $f\in \ell^{p_0}(\Z^d)\cap \ell^{p_1}(\Z^d)$ and each $t\in (0,1)$
\[
Tf\in \ell^{q_t}(\Z^d) \quad \textup{and} \quad  \|Tf\|_{q_t}\le M_t\|f\|_{p_t},
\]
where
\[
M_t=M_0^{1-t}M_1^t,\quad {1\/p_t}={1-t\/p_0}+{t\/p_1}, \quad \textup{and} \quad {1\/q_t}={1-t\/q_0}+{t\/q_1}.
\]
\end{theorem}

%\begin{proof} Repeat the continuous case, see %\cite[Theorem~IX.17]{RS75}.
%\end{proof}

\begin{theorem}[Discrete Young's inequality]
\lb{YE} Let  $f\in \ell^p(\Z^d)$,
$g\in \ell^s(\Z^d)$ and $h\in \ell^r(\Z^d)$, where
${1\/p}+{1\/s}+{1\/r}=2$ for some $p,s,r\ge 1$. Then
\begin{equation}
\lb{YE1} \Bigl|\sum_{n,m\in \Z^d}f_ng_{n-m}h_m\Bigr|\le \|f\|_p\|g\|_s\|h\|_r.
\end{equation}
\end{theorem}

%\begin{proof} Repeat the continuous case, see %\cite[Theorem~4.2]{LL97}.
%\end{proof}

We end with the following estimate.

\begin{lemma}\label{SummationEstimate} Let $\alpha>1$, $0< \beta< 1$ and $t\geq 1$. Then
\[
\sum_{n,m\in\Z}(1+|n|)^{-\alpha} (1+|m|)^{-\alpha} (1+ ||n-
m|-t|)^{-\beta}  \leq
\Bigl(\frac{2\alpha^2}{(\alpha-1)^2}+
\frac{4\alpha}{(\alpha-1)(p\alpha-1)^{\frac1{p}}} \Bigl(\frac{16}{1-\beta}\Bigr)^\frac1{r} \Bigr) t^{-\beta},
\]
where $p= (1+\beta)/(1-\beta)$ and $r= (1+\beta)/(2\beta)$.
\end{lemma}

\begin{proof}
Employing the estimate
\[
\one_{\left[\left||n-m|-t\right| <\tfrac{t}2\right]} \leq \one_{\left[|n-m|>\tfrac{t}2\right]} \leq \one_{\left[|n|>\tfrac{t}4\right]} +  \one_{\left[|m|>\tfrac{t}4\right]},
\]
we may simplify:
\[
\begin{aligned}
& \sum_{n,m\in\Z}(1+|n|)^{-\alpha} (1+|m|)^{-\alpha} (1+ ||n-
m|-t|)^{-\beta} \\
&\quad \leq
\sum_{n,m\in \Z}(1+|n|)^{-\alpha} (1+|m|)^{-\alpha} \one_{\left[\left| |n-m|-t\right|\geq \tfrac{t}2\right]}(1+ ||n-m|-t|)^{-\beta}\\
&\qquad + 2  \sum_{n,m\in \Z} \one_{\left[|n|>\tfrac{t}4\right]}(1+|n|)^{-\alpha}(1+|m|)^{-\alpha} \one_{\left[\left| |n-m|-t\right| < \tfrac{t}2\right]}(1+ ||n-m|-t|)^{-\beta}\\
&\quad \leq 2^\beta t^{-\beta}\Bigl(\sum_{n\in\Z}(1+|n|)^{-\alpha}\Bigr)^2\\
&\qquad + 2 \bigl\|\one_{\left[|n|> \tfrac{t}4\right]}(1+|n|)^{-\alpha}\bigr\|_p \bigl\|(1+ |m|)^{-\alpha}\bigr\|_1 \bigl\|\one_{\left[\left||w|-t\right|<\tfrac{t}2\right]}(1+||w|-t|)^{-\beta}\bigr\|_r,
\end{aligned}
\]
where we used the discrete Young inequality (Theorem~\ref{YE}) in the last step with $p = \frac{1+\beta}{1-\beta}$, $s=1$ and $r = \frac{1+\beta}{2\beta}$. Note that $r\beta = (1+\beta)/2< 1$.

To complete the proof we observe
\[\label{rho-alpha-norm}
\sum_{n\in\Z}(1+|n|)^{-\alpha}  = 1 + 2 \sum_{n=2}^\infty n^{-\alpha} \leq 1 + \int_1^\infty x^{-\alpha}\, dx = 1+ \frac1{\alpha-1} =\frac{\alpha}{\alpha-1},
\]
\[
\begin{aligned}
   \bigl\|\one_{\left[|n|> \tfrac{t}4\right]}(1+|n|)^{-\alpha}\bigr\|_p^p & = 2 \sum_{n=1}^\infty \one_{\left[n> \tfrac{t}4\right]}(1+n)^{-p\alpha}
\leq 2 \int_0^\infty  \one_{\left[x> \tfrac{t}4\right]}(1+x)^{-p\alpha}\, dx\\
&  = \frac2{p\alpha-1}\bigl(1+\frac{t}4\bigr)^{1-p\alpha} \leq
\frac2{p\alpha-1}\bigl(1+\frac{t}4\bigr)^{1-p}
\leq \frac{4^{p}}{2(p\alpha-1)} t^{1-p},
\end{aligned}
\]
and finally
\[\label{r-norm}
\begin{aligned}
& \bigl\|\one_{\left[\left||w|-t\right|<\tfrac{t}2\right]}(1+||w|-t|)^{-\beta}\bigr\|_r^r  = \sum_{w\in\Z}\one_{\left[\left||w|-t\right|<\tfrac{t}2\right]}(1+||w|-t|)^{-r\beta}\\
&\qquad  = 2 \sum_{w=1}^\infty \one_{\left[|w-t|<\tfrac{t}2\right]}(1+|w-t|)^{-r\beta}.
\end{aligned}
\]
Let $[c]$ denote the integer part of a real number $c>0$.   Write
$t = [t] + \delta$ with $0\leq \delta <1$
and estimate
\[
\begin{aligned}
&\sum_{w=1}^\infty \one_{\left[|w-t|<\tfrac{t}2\right]}(1+|w-t|)^{-r\beta}
 \leq \sum_{w= [t/2]+1}^{[\frac{3t}2]}(1+|w-t|)^{-r\beta}\\
 &\quad
= \sum_{k=[\frac{t}2]+1-[t]}^{[3t/2]-[t]} (1+|k-\delta|)^{-r\beta}
= \sum_{k=[t/2]-[t] +1}^0 (1+|k|)^{-r\beta} +\sum_{k=1}^{[3t/2]-[t]} (1+(k-1))^{-r\beta}\\
& \quad \leq  2 \sum_{k=0}^{[t/2]} (1+k)^{-r\beta}
\leq 2+  2 \int_0^{[\frac{t}2]} x^{-r\beta}\, dx
\leq 2+ \frac{2}{1-r\beta} (t/2)^{1-r\beta}\leq  \frac{6}{1-\beta} t^{1-r\beta}.
\end{aligned}
\]
Here we used that for $c>0$ we have $[2c]-[c] \leq [c]+1$ and
$[3c]-[2c]\leq [c]+1$. Inserting into \eqref{r-norm} we find that
\[
\bigl\|\one_{\left[\left||w|-t\right|<\tfrac{t}2\right]}(1+||w|-t|)^{-\beta}\bigr\|_r
\leq \Bigl(\frac{12}{1-\beta}\Bigr)^\frac1{r} t^{\frac1{r}-\beta}.
\]

In conclusion, recalling that $1/p + 1/r = 1$, we have
\[
\begin{aligned}
& \sum_{n,m\in\Z}(1+|n|)^{-\alpha} (1+|m|)^{-\alpha} (1+ ||n-
m|-t|)^{-\beta}\\
&\quad \leq \Bigl(\frac{2\alpha^2}{(\alpha-1)^2}+
\frac{4\alpha}{(\alpha-1)(p\alpha-1)^{\frac1{p}}} \Bigl(\frac{24}{1-\beta}\Bigr)^\frac1{r} \Bigr) t^{-\beta}.
\end{aligned}
\]

\end{proof}

\end{appendix}

\medskip

\footnotesize \no\textbf{Acknowledgments.} \footnotesize Various
parts of this paper were written during Evgeny Korotyaev's stay as a
VELUX Visiting Professor at the Department of Mathematics, Aarhus
University, Denmark. He is grateful to the institute for the
hospitality. In addition, our study was
supported by the RSF grant No 15-11-30007 and the Danish Council for Independent Research grant No 1323-00360.

\bibliographystyle{amsrefs}

\end{document}